\documentclass[aps,pra,twocolumn,superscriptaddress,longbibliography,nofootinbib]{revtex4-1}
\usepackage[normalem]{ulem}
\usepackage{float}
\usepackage{graphicx}  
\usepackage{dcolumn}          
\usepackage{amssymb}
\usepackage{appendix}
\usepackage{physics}   
\usepackage{mathtools}
\usepackage{esvect}
\usepackage{wrapfig}
\usepackage{amsthm}
\usepackage{verbatim}
\usepackage{bbm}
\usepackage[colorlinks=true,linkcolor=blue,citecolor=red,plainpages=false,pdfpagelabels]{hyperref}

\usepackage[mathscr]{euscript}

\def\Tr{\operatorname{Tr}}

\def\sq{\operatorname{sq}}

\def\supp{\operatorname{supp}}

\def\GHZ{\operatorname{GHZ}}

\def\id{\operatorname{id}}

\def\({\left(}
\def\){\right)}
\def\[{\left[}
\def\]{\right]}

\newcommand{\mc}[1]{\mathcal{#1}}

\newtheorem{theorem}{Theorem}
\newtheorem{note}{Note}

\newtheorem{corollary}{Corollary}

\newtheorem{definition}{Definition}

\newtheorem{lemma}{Lemma}
\newtheorem{observation}{Observation}

\newtheorem{proposition}{Proposition}
\newtheorem{remark}{Remark}

\def\>{\rangle}
\def\<{\langle}

\begin{document}

\widetext

\title{Fundamental limitations on the device-independent quantum conference key agreement}

\author{Karol Horodecki}
\affiliation{Institute of Informatics, National Quantum Information Centre, Faculty of Mathematics,
Physics and Informatics, University of Gda\'nsk, Wita Stwosza 57, 80-308 Gda\'nsk, Poland}
\affiliation{International Centre for Theory of Quantum Technologies, University of Gdańsk,
80-308 Gdańsk, Poland}
\affiliation{School of Electrical and Computer Engineering, Cornell University, Ithaca, New York 14850, USA}
\author{Marek Winczewski}
\affiliation{International Centre for Theory of Quantum Technologies, University of Gdańsk,
80-308 Gdańsk, Poland}
\affiliation{Institute of Theoretical Physics and Astrophysics, National Quantum Information Centre,
Faculty of Mathematics, Physics and Informatics, University of Gda\'nsk, Wita Stwosza 57, 80-308 Gda\'nsk, Poland}
\author{Siddhartha Das}
\affiliation{Centre for Quantum Information and  Communication (QuIC), \'{E}cole polytechnique de Bruxelles,   Universit\'{e} libre de Bruxelles, Brussels, B-1050, Belgium}
\affiliation{Center for Security, Theory and Algorithmic Research, International Institute of Information Technology, Hyderabad, Gachibowli, Telangana 500032, India}

\date{\today}
\begin{abstract}
We provide several general upper bounds on the rate of a key secure against a quantum adversary in the device-independent conference key agreement (DI-CKA)  scenario. They include bounds by reduced entanglement measures and those based on multipartite secrecy monotones such as a multipartite squashed entanglement-based measure, which we refer to as reduced c-squashed entanglement. We compare the latter bound with the known lower bound for the protocol of conference key distillation based on the parity Clauser-Horne-Shimony-Holt  game. We also show that the gap between the DI-CKA rate and the device-dependent rate is inherited from the bipartite gap between device-independent and device-dependent key rates, giving examples that exhibit the strict gap.
\end{abstract}

\maketitle

\section{Introduction}
Building a quantum secure  internet is one of the most important challenges in the field of quantum technologies \cite{DM03,Wehner2018}.
It would ensure worldwide information-theoretically secure communication.
The idea of quantum repeaters \cite{Dr1999,Muralidharan2016,ZXC+18} gives hope that this dream will come true.
However, the level of quantum security proposed originally in a seminal article by Bennett and Brassard \cite{BB84} seems to be insufficient due to the fact that on the way between an honest manufacturer and an honest user, an active hacker can change with the inner workings of a quantum device, making it totally insecure \cite{Pironio2009}. Indeed, the hardware Trojan-horse attacks on random number generators are known \cite{Becker2013}, and the active hacking on quantum devices became a standard testing approach since the seminal attack by Makarov \cite{Makarov2009}. The idea of device-independent (DI) security overcomes this obstacle \cite{E91,Pironio2009} (see also \cite{Bell-nonlocality} and references therein).
Although difficult to be done in practice, it has been demonstrated quite recently in several recent experiments \cite{Harald_exp,Renner_exp,JWP_exp}. 

In parallel, the study of the limitations of this approach in terms of upper bounds on the distillable key has been put forward \cite{KW17,CFH21,FBL+21,KHD21}. However, these approaches focus on point-to-point quantum device-independent secure communication. In this paper  we introduce the upper bounds on the performance of the device-independent conference key agreement (DI-CKA) \cite{Murta2020,RMW18}. The task of the conference agreement is to distribute to $N>2$ honest parties the same secure key for one-time-pad encryption. 
A protocol achieving this task in a device-independent manner has been shown in Ref.~\cite{RMW18}.  We set an upper bound on the performance of such protocols in a network setting. 

We focus on physical behaviors with $N$ users (for arbitrary $N>2$), where each user is both the sender and receiver of the behavior treated as a black box. This situation is a special case of a network describable with a multiplex quantum channel where inputs and outputs are classical with quantum phenomena going inside the physical behavior~\cite{DBWH19}. All $N$ trusted parties have the role of both the sender to and receiver from the channel and their goal is to
obtain a secret key in a device-independent way against a quantum adversary. Aiming at upper bounds on the device-independent key, we narrow the consideration to the independent and identically distributed case. In this scenario, the honest parties share $n$ identical devices.
All the $N$ parties set (classical) inputs ${\bf x}=(x_1,\ldots  ,x_N)$ to each of the $n$ shared devices $P({\bf a}|{\bf x})$ and receive (classical) outputs ${\bf a}=(a_1,\ldots  ,a_N)$ from each of them. We restrict our consideration to quantum devices. Such devices are realized by certain measurements ${\cal M}\equiv \otimes_{i=1}^N {\mathrm M}^{x_i}_{a_i}$ on quantum states $\rho_{A_1,\ldots  ,A_N}\equiv \rho_{N(A)}$. We define these devices $(\rho_{N(A)},{\cal M})=\Tr \left[\rho_{N(A)} (\otimes_{i=1}^N {\mathrm M}^{x_i}_{a_i}) \right]$. 

In this work we provide upper bounds on the device-independent conference key distillation rates for arbitrary multipartite states. As the first main result, we introduce a multipartite generalization of the {\it cc-squashed entanglement} provided in Ref.~\cite{AL20} and developed in Ref.~\cite{KHD21}. With a little abuse of notation with respect to that used in Refs.  \cite{CFH21,KHD21}, for the sake of the reader, we will omit the fact that the measure is multipartite as well as reduce the abbreviation $cc$ in its name and here call it
just reduced c-squashed entanglement, denoting it by $E_{sq,dev}^{c}$. We show that $E_{sq,dev}^{c}$ upper bounds the device-independent key rate in the independent and identically distributed setting, achieved by protocols which use a single input to generate the key $\hat{\bf{x}}$, denoted by $K_{DI,dev}^{iid,{\hat{\bf x}}}$. The subscript dev in the notation refers to the fact that the adversary has to mimic the statistics of the honestly implemented device. We then generalize this to the case when only some parameters of the device have to be reproduced by the attack and refer to quantities with the subscript par. Typical parameters are the level of violation of a Bell inequality and the quantum bit error rate. In the above finding, we use the notion of multipartite squashed entanglement given in Ref.~\cite{Yang2009}. Therein, the abbreviation c stands for classical as the systems of the honest parties are classical due to the measurement accordingly to the definition. The bound reads
\begin{align}
&K_{DI,dev}^{iid,\hat{{\bf x}}}(\rho_{N(A)},{\cal M})\leqslant E_{sq,dev}^{c}(\rho_{N(A)},{\cal M})\nonumber\\
&\equiv\inf_{(\sigma_{N(A)},{\cal N})=(\rho_{N(A)},{\cal M})} \nonumber \\
&\frac{1}{ N-1}I(A_1: \ldots  :A_N\downarrow E)_{{\cal N}(\hat{\bf x})\otimes\mathbbm{1}\sigma_{N(A)}}.
\end{align}
In the above $\downarrow$ denotes the action of any channel transforming $E$ to some system $E^\prime$ and $I(A_1:\ldots  :A_N|E)_{\sigma_{N(A)}}= \sum_{i=1}^{N} S(A_i|E)_{\sigma_{N(A)}} - S(A_1,\ldots ,A_N|E)_{\sigma_{N(A)}}$, with $S(X|Y)_{\sigma_{X,Y}}$ the conditional von Neumann entropy of the state $\sigma_{X,Y}$. The quantity $I(A_1:\ldots: A_N|E')$ (after the action of the channel on $E$) is evaluated on the classical-quantum state emerging from the measurement ${\mathrm N}^{\hat{\bf{x}}}_{\bf{a}}$ corresponding to input $\hat{\bf x}$ of the device $(\sigma_{N(A)},{\cal N})$, on systems $A_1,\ldots,  A_N$. Let us note here that, due to the findings of Ref. \cite{KHD21}, for $N=2$ and the case when the system $E^\prime$, as the output of a channel, is classical, the above bound is equal to the intrinsic information given in Ref. \cite{FBL+21} for the case of a single measurement generating the key.
\begin{figure}[t]
    \center{\includegraphics[trim={10.5cm 9.5cm 8cm 5cm},clip,width=10.5cm]{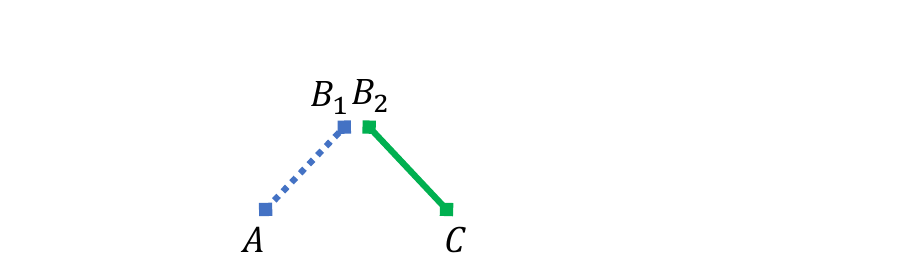}}
    \caption{Depiction of a construction of a tripartite state $\rho_{A(B_1B_2)C}$ with a gap $K_{DI}(\rho_{A(B_1B_2)C}) < K_{DD}(\rho_{A(B_1B_2)C})$ from a state
    $\rho_{AB_1}$ (blue dotted line) satisfying $K^{\downarrow}(\rho_{AB_1})<K_{DD}(\rho_{AB_1})$ (provided in Ref.~\cite{CFH21}) and any state     $\rho_{B_2C}$ (green solid line) satisfying $K_{DD}(\rho_{B_2C})	\geqslant K_{DD}(\rho_{AB_1})$, e.g., the singlet state.
    }
    \label{fig:multi}
\end{figure}

All quantum states considered in this paper are $N$-partite unless it is stated otherwise. Therefore, for the sake of the conciseness of the notation we omit the subscript $N(A)$ in some places.

Our technique is based on the approach of Ref.~\cite{CEHHOR}, where the upper bounds on a key distillable against a quantum adversary via  local operations and public communication (LOPC) were studied. To achieve this, we generalize the upper bound via (quantum) intrinsic information to the case in which the adversary's system can be of infinite dimension. This technical contribution was necessary, as in the case of a device-independent attack, the dimension of the attacking state can be infinite. Indeed, while measurements $\bf x$ produce from the attacking state $\sigma$ finite-dimensional results $\bf a$ yielding a quantum behavior $P({\bf a }|{\bf x})=\Tr \left[ \sigma {\mathrm M}^{{\bf x}}_{{\bf a}} \right]$, the system of the adversary which may hold purification of the state $\sigma$, can still be of infinite dimension.\footnote{In that we have filled in the gap in the proofs of Corollaries 3 and
4 of Ref. \cite{KHD21}, where implicit assumption of the adversary holding
finite-dimensional state has been made.} 

\begin{figure}[t]
    \center{\includegraphics[trim=0cm 0cm 0cm 0cm,width=8.5cm]{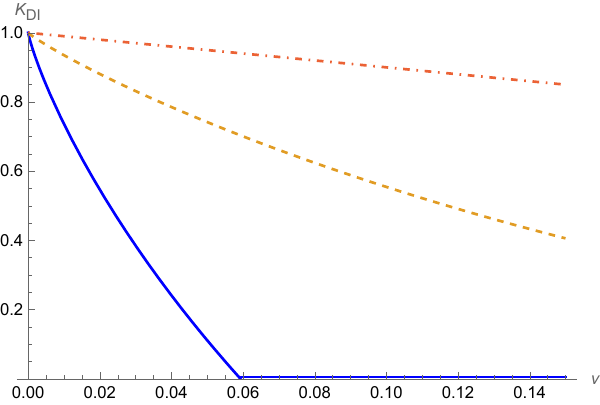}}
    \caption{Plot of upper and lower bounds on the DI-CKA of Ref.~\cite{RMW18}. The yellow dashed line represents an upper bound (not fully optimized) on the upper bound $\frac{1}{N-1}I(N(A)\downarrow E)$ from Eq.~\eqref{eqn:UBmutual}  with the attack strategy in Eq.~\eqref{eqn:attack1}. The red dash-dotted curve is the trivial upper bound obtained in Corollary \ref{cor:trivial} via the relative entropy of entanglement bound ($1-\nu$). The blue solid line represents the lower bound from Ref.~\cite{RMW18}.}
    \label{fig:attack2}
\end{figure}

We then compare the obtained upper bounds with the lower bound on the DI-CKA provided in Ref.~\cite{RMW18}  (see Fig.~\ref{fig:attack2}). We obtain the plot by considering simplification of the reduced c-squashed entanglement. Namely, the extension to the adversarial system of the state attacking the honest parties device is classical, i.e., diagonal in the computational basis. For that reason, the bound which we use is in fact a {\it secrecy monotone}, called (multipartite) {\it intrinsic information}~\cite{Cerf2002}.

As the second main result, we show how to construct multipartite states with a strict gap between the rate of the quantum device-independent conference key rates $K_{DI}$ and quantum device-dependent conference key rates $K_{DD}$. As a proxy, we use a bipartite state that satisfies $K_{DD}>K^{\downarrow}$, where $K^{\downarrow}$ is the {\it reduced distillable key} introduced in Ref.~\cite{CFH21}. The reduced distillable key of $\rho_{N(A)}$ is the maximum value of the choice of measurements ${\cal M}$ of the distillable key of the adversarial state $\sigma_{N(A)}$ minimized over the 
choices of the state $\sigma_{N(A)}$ and measurements ${\cal N}$ so that the device $(\sigma_{N(A)},{\cal N})\equiv \left\{\Tr \left[ \sigma_{N(A)} (\otimes_{i=1}^N {\mathrm N}^{x_i}_{a_i} )\right] \right\}_{\bf{a}|\bf{x}}$, is equal to the honestly implemented device $(\rho, {\cal M})$. The mentioned gap between $K_{DI}$ and $K_{DD}$ means that also in a multipartite case for some states $\rho_{A_1,\ldots  ,A_N}$ (and any $N> 2$), there is neither a Bell-like inequality that can be used for testing nor a distillation protocol based on LOPC  that can achieve $K_{DI}(\rho_{A_1,\ldots , A_N})=K_{DD}(\rho_{A_1,\ldots , A_N})$. See Fig.~\ref{fig:multi}.

Finally, we discuss the issue of genuine nonlocality~\cite{Bell-nonlocality} and genuine entanglement in the context of the DI-CKA~\cite{Horodecki2009}. As the third main result, we provide a non-trivial bound on the device-independent key achievable in a parallel measurement scenario, when all
the parties set all values of the inputs $x_i$ in parallel. Furthermore, generalizing reduced bipartite entanglement measures in Ref.~\cite{KHD21} to multipartite entanglement measures, we show that the reduced regularized relative entropy of genuine entanglement~\cite{DBWH19} upper bounds the DI-CKA rate of multipartite quantum states. We further focus on the performance of protocols using a single input for key generation, as using such protocols is standard practice (see, e.g., \cite{AFFRV18}).

The remainder of this paper is organized as follows.
Section \ref{sec:di-qkd-state} is devoted to basic facts and provides bounds on the DI conference key via entanglement measures.  In Section~\ref{sec:multi_squashed}, we develop an upper bound on the {DI-CKA} via reduced c-squashed entanglement. In Section~\ref{sec:with_figure}, we provide particular examples for the performance of upper bounds considered in the paper. In Section~\ref{sec:gap}, we provide examples of multipartite states which
exhibit a fundamental gap between the device-dependent and -independent secure key rates. Section~\ref{sec:GNLandENT} discusses the connection between genuine nonlocality, entanglement, and DI-CKA. We conclude in Section~\ref{sec:discussion} with a summary and some directions for future study.

\begin{note}\label{note1}
Theorem 7 of Ref.~\cite{DSW18} states that the two multipartite entanglement measures, multipartite squashed entanglement $E_{sq}$ (Definition~4 in~\cite{HWD22}) and its dual $\widetilde{E}_{sq}$ (Definition~7 in~\cite{HWD22}), of a multipartite state are the same. As a consequence, the reduced c-squashed entanglement $E^c_{sq}(\rho,{\rm M})$ (Definition~5 in~\cite{HWD22}) and the dual c-squashed entanglement $\widetilde{E}^c_{sq}(\rho,{\rm M})$ (Definition~8 in~\cite{HWD22}) of a device $(\rho,{\rm M})$ are also the same. For this reason, Section~IV in the published version of the manuscript [Phys.~Rev.~A 105(2), 022604 (2022)]~\cite{HWD22} can be skipped in reading.
\end{note}

\begin{note}
    The current version differs from [Phys.~Rev.~A 105(2), 022604 (2022)]~\cite{HWD22} by removing Section IV, references to it, along with including updated Fig.~\ref{fig:attack2} and fixing typographic mistakes in notation and explanation of Eq.~\eqref{eqn:attack1}. There is no change with respect to the results and proofs of the [Phys.~Rev.~A 105(2), 022604 (2022)]~\cite{HWD22}.
\end{note}

\section{Bounds on device-independent key distillation rate of states}\label{sec:di-qkd-state}

In this section, we introduce the scenario of
device-independent conference key distillation from $n$ identical devices, each shared by $N$ honest users. 
We then introduce definitions and facts used in subsequent sections.

Consider a setup wherein $N$ multiple trusted spatially separated users (allies) have to extract a secret key, i.e., conference key, against the quantum adversary. Since we aim
at upper bounds on the device-independent conference key, we assume that the parties
share $n$ identical devices. The device
has its honest implementation, which 
is reflected by the state and measurement, denoted by $(\rho,{\cal M})$, 
that were intended to be delivered by a provider. The adversary may replace 
this honest implementation with a different
device $(\sigma,{\cal N})$, however, such that it yields the same {\it input-output statistics} as the honest one. Typically, the statistics tested by the allies are the level of violation of some Bell inequality, and the quantum bit error rate, i.e., the probability that the outputs of the honest parties are not equal to each other given
the raw key has been generated. In some cases, we will also consider the full statistics reflected by the pair $(\rho,{\cal M})$. We note here
that the state $\sigma$ can be finite or infinite-dimensional, as we do not restrict the strategies of the adversary in that respect.

The honest device is given by  $\mc{M}\equiv \{M^{x_1}_{a_1}\otimes M^{x_2}_{a_2}\otimes\ldots\otimes M^{x_N}_{a_N}\}_{\textbf{a}|\textbf{x}}$, where $ \textbf{x}\coloneqq (x_1,x_2,\ldots , x_N)$ and $\textbf{a}\coloneqq (a_1,a_2,\ldots ,a_N)$, for some $N\in\mathbb{N}$. For each $i\in [N]\coloneqq \{1,2,\ldots , N\}$, the set $\{a_i\}$ denotes the finite set of measurement outcomes for measurement choices $x_i$. The measurement outcomes, i.e., outputs of the device, are secure from the adversary and assumed to be in the possession of the receivers (allies). The joint probability distribution is given as 
\begin{equation}
   p({\bf a}|{\bf x})=\Tr[M^{x_1}_{a_1}\otimes M^{x_2}_{a_2}\otimes\ldots\otimes M^{x_N}_{a_N} \rho_{A_1A_2\ldots  A_N}] 
\end{equation}
for measurement $\mc{M}$ on $N$-partite state $\rho_{N(A)}$ defined on the separable Hilbert space $\mc{H}_{A_1}\otimes\mc{H}_{A_2}\otimes\ldots\mc{H}_{A_N}$; in what follows we will use $N(A)\equiv A_1\ldots  A_N$ for the ease of notation. The tuple $\left\{\rho,{\rm M}\right\}$, where ${\rm M}\coloneqq \left(\left\{M^{x_1}_{a_1}\right\}_{x_1},\left\{M^{x_2}_{a_2}\right\}_{x_2},\ldots , \left\{M^{x_N}_{a_N}\right\}_{x_N}\right)$, is called the quantum strategy of the distribution. The number of inputs $\{x_i\}$ and corresponding possible outputs $\{a_i\}$ of the local measurement at $A_i$ are arbitrarily finite in general. We denote the identity superoperator by $\id$ and the identity operator by $\mathbbm{1}$. 

Let $\omega(\rho, \mc{M})$ denote the violation of the given multipartite Bell-type inequality $\mc{B}$ by state $\rho$ when the measurement settings are given by $\mc{M}$. We note that by multipartite Bell-type inequality we mean any inequality derived using locally realistic hidden variable (LRHV) theories~see, e.g.,~\cite{Mer90,Ard92,BK93,SS02,ZB02,WW02,YCLO12,HSD15,Luo21}) such that any violation of a given inequality by a density operator implies the non-existence of an LRHV model for the device represented by this state and some measurements. There are families of Bell-type inequalities directly based on the joint probability distribution of local measurements that get violated by all pure multipartite (genuinely) entangled states~\cite{YCLO12,HSD15}. On the other hand, there are Bell-type inequalities based on correlation functions of local measurements for which some families of pure multipartite (genuinely) entangled states satisfy the inequalities~\cite{ZBLW02}.

Let $P_{err}(\rho,\mc{M})$ denote the expected quantum bit error rate (QBER). Both the Bell violation and the QBER are functions of the probability distribution of the behavior. In addition, $\Phi^{{\rm GHZ}}_{\vv{N}}\coloneqq \op{\Phi^{{\rm GHZ}}_{\vv{N}}}$ denotes the $N$-partite Greenberger--Horne--Zeilinger (GHZ) state.
\begin{equation}
    \ket{\Phi^{{\rm GHZ}}_{\vv{N}}}=\frac{1}{\sqrt{d}}\sum_{i=0}^{d-1}\ket{i}_{A_1}\otimes\ket{i}_{A_2}\otimes\cdots\otimes\ket{i}_{A_N}
\end{equation}
for $d=\min_{i}\dim(\mc{H}_{A_i})$. For $N=2$, $\Phi^{{\rm GHZ}}_{\vv{2}}$ is a maximally entangled (Bell) state $\Phi_{\vv{2}}$.

If $\{p({\bf a}\vert {\bf x})\}_{{\bf a}|{\bf x}}$ obtained from $(\rho,\mc{M})$ and another pair of states and measurements $(\sigma,\mc{N})$ are the same, we write $(\sigma,\mc{N})= (\rho,\mc{M})$. In most DI-CKA protocols, instead of using the statistics of the full correlation, we use the Bell violation and the QBER to test the level of security of the observed statistics. In this way, for practical reasons, the protocols coarse grain the statistics, and we only use partial information of the full statistics to extract the device-independent key. 
In this context, the notation $(\sigma,\mc{N})= (\rho,\mc{M})$ also implies that $\omega(\sigma,\mc{N})= \omega(\rho,\mc{M})$ and $P_{err}(\sigma,\mc{N})=P_{err}(\rho,\mc{M})$. When conditional probabilities associated with $(\rho,\mc{M})$ and $(\sigma,\mc{N})$ are $\varepsilon$-close to each other, then we write $(\rho,\mc{M})\approx_{\varepsilon}(\sigma,\mc{N})$. For our purpose, it suffices to consider the distance
\begin{equation}
    d(p,p')=\sup_{{\bf x}}\norm{p(\cdot\vert {\bf x})-p'(\cdot \vert {\bf x})}_1\leq \varepsilon.
\end{equation}

The device-independent distillable key rate of a device is informally defined as the supremum of the finite key rates $\kappa$ achievable by the best protocol on any device compatible with $(\rho,\mc{M})$, within an appropriate asymptotic blocklength limit and security parameter. Another approach taken is to minimize the key rate of the statistics compatible with the Bell parameter and a QBER (see, e.g.,~\cite{AL20}). For our purpose, we constrain ourselves to the situation when the compatible devices are supposedly independent and identically distributed. This constraint is because, as noted in Ref. \cite{CFH21}, the upper bound on the key in the independent and identically distributed scenario is automatically the upper bound on the device-independent conference key in the general scenario since the independent and identically distributed attack is just one of the possible attacks in the general device-independent scenario.

An ideal conference key state $\tau^{(K)}$, with $\log_2 K$ secret key bits for $N$ allies, is 
\begin{align}
\tau^{(K)}_{N(A)E}\coloneqq \frac{1}{K}\sum_{k=0}^{K-1}\op{k}_{A_1}\otimes \op{k}_{A_2}\otimes\cdots \otimes \op{k}_{A_N}\otimes \tau_E, \label{eq:key-state}
\end{align}
where $\tau_E$ is a state of the only system $E$ accessible to an adversary, i.e., the adversary is uncorrelated with trusted users and gets no information about their secret bits. Consider the relations
\begin{align}
    (\rho,\mc{M}) &\approx_\varepsilon  (\sigma,\mc{N}) \label{eq:k-1},\\
    \omega(\rho,\mc{M})& \approx_{\varepsilon} \omega(\sigma,\mc{N})\label{eq:k-2},\\
    P_{err}(\rho,\mc{M})& \approx_{\varepsilon} P_{err}(\sigma,\mc{M})\label{eq:k-3},
\end{align}
where Eq.~\eqref{eq:k-1} implies Eqs.~\eqref{eq:k-2} and ~\eqref{eq:k-3}.

Formally, the definition of device-independent quantum key distillation rate in the independent and identically distributed scenario is given as follows. 
\begin{definition}[cf.~\cite{CFH21}]\label{def:di-dev}
The maximum (multipartite) device-independent quantum key distillation rate of a device $(\rho,\mc{M})$ with independent and identically distributed behavior is defined as
\begin{equation}
    K^{iid}_{DI,dev}(\rho,\mc{M})\coloneqq \inf_{\varepsilon>0}\limsup_{n\to \infty} \sup_{\hat{\mc{P}}} \inf_{\eqref{eq:k-1}} \kappa^\varepsilon_{n} \left(\hat{\mc{P}}\left((\sigma,\mc{N})^{\otimes n}\right)\right),
\end{equation}
where $\kappa_n^\varepsilon$ is the rate of a key distillation protocol $\hat{\mc{P}}$ producing $\varepsilon$-secure output, acting on $n$ copies of the state $\sigma$, measured with $\mc{N}$. Here $\hat{\mc{P}}$ is a protocol composed of classical local operations and public (classical) communication (CLOPC) acting on $n$ identical copies of $(\sigma,\mc{N})$ which, composed with the measurement, results in a quantum local operations and public (classical) communication (QLOPC) protocol.
\end{definition}

The following lemma follows from the definition of $K^{iid}_{DI}$ (generalizing statements from bipartite DI quantum key distillation in Refs.~\cite{CFH21,KHD21} to the DI-CKA).
\begin{lemma}\label{lem:di-state}
The maximum (multipartite) device-independent quantum key distillation rate $K^{iid}_{DI}$ of a device $(\rho,\mc{M})$ is equal to the maximum (multipartite) device-independent quantum key distillation rate of a device $(\sigma,\mc{N})$ when $(\rho,\mc{M})= (\sigma,\mc{N})$:
\begin{align}
& (\rho,\mc{M})=(\sigma,\mc{N}) \implies K^{iid}_{DI,dev}(\rho,\mc{M})= K^{iid}_{DI,dev}(\sigma,\mc{N}).\label{eq:p-d-i}
\end{align}
\end{lemma}

The maximal DI-CKA rate $K_{DI,dev}(\rho, \mathcal{M})$ for the device $(\rho,\mc{M})$ is upper bounded by the maximal device-dependent conference key agreement (DD-CKA) rate $K_{DD}(\sigma)$ for all $(\sigma,\mc{N})$ such that $(\sigma,\mc{N})= (\rho,\mc{M})$ (cf.~\cite{CFH21}), i.e.,
\begin{equation}
    K_{DI,dev}(\rho, \mathcal{M})\leq \inf_{(\sigma,\mc{N})= (\rho,\mc{M})} K_{DD}(\sigma)\label{eq:di-bound10}.
\end{equation}

The device-dependent quantum key distillation rate $K_{DD}(\rho)$ (cf.~\cite{CEHHOR}) is the maximum secret key (against quantum eavesdropper) that can be distilled between allies using local operations and classical communication (LOCC) (see, e.g.,~\cite{DBWH19}),  
\begin{align}\label{eq:kdd}
    &K_{DD}(\rho) \coloneqq \nonumber \\
    &\inf_{\epsilon > 0} \lim_{n\to \infty} \sup_{\Lambda \in QLOPC} \left\{ \left.\frac{\log_2 d_n}{n} ~\right|~ \Lambda ((\psi^{\rho})^{\otimes n}) \approx_\epsilon \tau^{(\log d_n)} \right\} ,
\end{align}
where $\psi^{\rho}$ is a purification of $\rho$ and 
 $\rho \approx_\epsilon \sigma \Longleftrightarrow \frac{1}{2}\norm{\rho-\sigma}_1 \le \epsilon$.

\begin{corollary}~\label{obs:ent-mes}
For entanglement measures $Ent$ which upper bound the maximum device-dependent key distillation rate, i.e., $K_{DD}(\rho)\leq Ent (\rho)$ for a density operator $\rho$, we have
\begin{align}\label{eq:di-e-1}
     K_{DI,dev}(\rho, \mathcal{M}) &\leq \inf_{(\sigma,\mc{N})= (\rho,\mc{M})} K_{DD}(\sigma) \\
&     \leq  \inf_{(\sigma,\mc{N})= (\rho,\mc{M})} Ent (\sigma).
\label{eq:FCHbound}
\end{align}
\end{corollary}

\begin{remark}
We do not make any assumption about the dimension of the Hilbert space on which the state $\sigma$ (Definitions~1--3) is defined as the systems can be finite-dimensional or infinite dimensional. We only assume that the systems $A_i$ accessible by the allies for key distillation upon measurement are finite dimensional in an honest setting. The systems $E$ accessible to an adversary can be finite dimensional or infinite dimensional, depending on the cheating strategy (see Lemma~\ref{lem:cehhor}, Appendix \ref{sec:app:B}). 
\end{remark}

As discussed above, a large class of device-independent quantum key distillation protocols relies on the Bell violation and the QBER of the device $p({\bf a}|{\bf x})$. For such protocols, we can define the device-independent key distillation protocol as follows.

\begin{definition}[cf.~\cite{AL20}]\label{def:di-par}
The maximal (multipartite) device-independent quantum key distillation rate of a device $(\rho,\mc{M})$ with independent and identically distributed behavior, Bell violation $\omega(\rho,\mc{M})$, and QBER $P_{err}(\rho,\mc{M})$ is defined as
\begin{align}
    & K^{iid}_{DI,par}(\rho,\mc{M})\nonumber \\
    &\quad\coloneqq \inf_{\varepsilon>0}\limsup_{n\to \infty} \sup_{\hat{\mc{P}}} \inf_{\eqref{eq:k-2},\eqref{eq:k-3}} \kappa^\varepsilon_{n}\left(\hat{\mc{P}}\left((\sigma,\mc{N})^{\otimes n}\right)\right).
\end{align}
\end{definition}
\begin{remark}
As Eq.~\eqref{eq:k-1} implies Eqs.~\eqref{eq:k-2} and \eqref{eq:k-3}, it follows from the definitions of $K^{iid}_{DI}(\rho,\mc{M})$ and $K^{iid}_{DI,par}(\rho,\mc{M})$ that $K^{iid}_{DI,par}(\rho,\mc{M})\leq K^{iid}_{DI,dev}(\rho,\mc{M})$.
\end{remark}

\begin{remark}[cf.~\cite{Bell-nonlocality}]
We note that there may exist states $\rho$ for which $K^{iid}_{DI}(\rho)=0$ but $K^{iid}_{DI}(\rho^{\otimes k})> 0$ for some $k\in\mathbbm{N}$.
\end{remark}

\section{Reduced c-squashed entanglement bound}
\label{sec:multi_squashed}
In this section we generalize the notion of the  { cc-squashed entanglement}~\cite{AL20,KHD21} to the multipartite form. Next we prove that the properly scaled reduced c-squashed entanglement serves as an upper bound on the device-dependent conference key of the classical-quantum state. Furthermore, via Lemma \ref{lem:convex} and Proposition \ref{prop:convex} we prove that the reduced c-squashed entanglement is convex. This result may be further applied to generate numeric upper bounds with the {\it convexification technique}~\cite{WDH22,KHD21}. Then we prove the main result of this section. Namely, we prove that the {\it independent and identically distributed quantum device-independent conference key} is upper bounded by the reduced c-squashed entanglement. Finally, we show that similar results hold when the honest parties broadcast the inputs to their devices so that the adversary can learn them.

In what follows, we first prove that the ``measured" version of the multipartite squashed entanglement $E_{sq}^{q}$ defined in Ref.~\cite{Yang2009}, if properly scaled,
upper bounds the conference key secure against the quantum
adversary. Let us first recall facts and definitions.
The {\it multipartite conditional mutual information} of a state $\rho_{A_1,\ldots,  A_NE}$ reads~\cite{Wat60}
\begin{align}\label{eq:i_multipartite}
    I(A_1:\ldots  :A_N|E)_{\rho}=\sum_{i=1}^{N} S(A_i|E)_{\rho} - S(A_1,\ldots  ,A_N|E)_{\rho}.
\end{align}
Here, the conditional entropy $S(A_i|E)_{\rho}=S(A_iE)_{\rho}-S(E)_{\rho}$, with $S(AB)\coloneqq -\Tr[\rho_{AB}\log_2  \rho_{AB}]$ and $S(A)\coloneqq -\Tr[\rho_A\log_2 \rho_A]$ being von Neumann entropies. The von Neumann entropy reduces to the Shannon entropy $H(X)$ for classical register $X$, $H(X)= -\sum_{x}p(x)\log_2 p(x)$, where $\{p(x)\}_{x}$ is the probability distribution associated with the random variable $X$. It will be crucial to note that the following identity 
holds \cite{Yang2009}
\begin{align}
    I(A_1:\ldots  :A_N|E)_{\rho} &=I(A_1:A_2|E)_{\rho} + I(A_3:A_1A_2|E)_{\rho}+\nonumber\\
    I(A_4:A_1A_2A_3|E)_{\rho} &+\cdots  +I(A_N:A_1
    \ldots  A_{N-1}|E)_{\rho}.
    \label{eq:expansion}
\end{align}
Here $I(A:B|C)_{\rho}=S(AC)_{\rho}+S(BC)_{\rho}-S(C)_{\rho}-S(ABC)_{\rho}$ is the {\it conditional mutual information}.
\begin{remark}\label{rem:permutations} Let $\Sigma$ be any permutation of indices $1, \ldots , N$. Then
\begin{align}
    &I(A_1:\ldots  :A_N|E)_{\rho}=I(A_{\Sigma(1)}:\ldots  :A_{\Sigma(N)}|E)_{\rho}\nonumber\\
    &=
    I(A_{\Sigma(1)}:A_{\Sigma(2)}|E)_{\rho} + I(A_{\Sigma(3)}:A_{\Sigma(1)}A_{\Sigma(2)}|E)_{\rho}\nonumber\\
    &+I(A_{\Sigma(4)}:A_{\Sigma(1)}A_{\Sigma(2)}A_{\Sigma(3)}|E)_{\rho}\nonumber\\ &+\ldots  +I(A_{\Sigma(N)}:A_{\Sigma(1)}
    \ldots  A_{\Sigma(N-1)}|E)_{\rho}.
\end{align}
\end{remark}

Further, the multipartite squashed  entanglement of 
a quantum state $\rho_{A_1,\ldots,  A_N}$ is defined as follows (Definition 3 of Ref.~\cite{Yang2009}).
\begin{definition} [\cite{Yang2009}]\label{def:CC_primal_ext}
For an $N$-partite state $\rho_{A_1,\ldots  ,A_N}$,
\begin{equation}
    E^{q}_{sq}(\rho_{A_1,\ldots,  A_N}):=\inf_{\sigma} I(A_1:A_2:\ldots  :A_N|E)_{\sigma},\label{eqn:def:CC_primal_ext}
\end{equation}
where the infimum is taken over states $\sigma_{A_1,\ldots,  A_NE}$ that are extensions of $\rho_{A_1,\ldots,  A_N}$, i.e., $\Tr_E[\sigma_{A_1,\ldots , A_NE}] = \rho_{A_1,\ldots , A_N}$. 
\end{definition}
We will need to generalize the notion of the {\it cc-squashed entanglement}~\cite{AL20,KHD21} to the multipartite form.
\begin{definition}\label{def:CC_primal_channel}
A reduced c-squashed entanglement of a state
$\rho_{A_1,\ldots , A_N}$ is defined as
\begin{align}
    &E^{c}_{sq}(\rho_{A_1,\ldots,  A_N},{\mathrm M}) \nonumber \\
    &:= \inf_{\Lambda:~{E\rightarrow E'}}I(A_1:\ldots  :A_N|E')_{{\mathrm M}_{N(A)}\otimes\Lambda \psi^{\rho}_{N(A) E}},\label{eqn:def:CC_primal_channel}
\end{align}
where ${\mathrm M}_{N(A)}$ is an $N$-tuple of positive-operator-valued measures (POVMs) ${\mathrm M}_{A_1},\ldots , {\mathrm M}_{A_N}$ and state $\ket{\psi^{\rho}_{N(A)E}}$ is a purification of $\rho_{N(A)}$.
\end{definition}
The first theorem comes with the following fact, which
is a multipartite generalization of Theorem~$5$ from Ref.~\cite{KHD21}, where $N=2$. Namely, the device-dependent conference key of the classical-quantum state is upper bounded by the properly scaled reduced c-squashed entanglement. We are ready to state a theorem that shows that the c-squashed entanglement, when properly scaled, upper bounds the device-dependent key.

\begin{theorem}
For an $N$-partite state $\rho_{N(A)}$, its purification $\psi^\rho_{N(A)E}$, and an $N$-tuple of POVMs ${\mathrm M}_{N(A)}$, there is
\begin{equation}
K_{DD}({\mathrm M}_{N(A)}\otimes \id_E\psi^\rho_{N(A)E}) \leq \frac{1} {N-1}E^{c}_{sq}(\rho_{N(A)},{\mathrm M_{N(A)}}).
\end{equation}
\label{thm:bound_on_ddkey}
\end{theorem}
\begin{proof}
We closely follow the proof of Theorem~$3.5$ of Ref.~\cite{CEHHOR}, however, based not on Theorem~$3.1$ of Ref.~\cite{CEHHOR}, but its generalization to
the case where system $E$ need not be finite
(see Lemma \ref{lem:infinite_dim}, Appendix \ref{sec:app:B}). 
Namely, any function which satisfies  
(i) monotonicity under LOPC, (ii) asymptotic continuity, (iii) normalization, and (iv) subadditivity is, after regularization, an
upper bound on the distillable key secure against the quantum adversary ($K_{DD}$). 

We first show the monotonicity. The LOPC consist of
local operation and 
public communication.
A local operation consists
of adding a local ancilla,
performing a unitary transformation, and a partial trace. It is easy to see
that adding a local ancilla
at one system does not 
alter this quantity. The same holds for the unitary transformation. The partial trace does not increase it as it can be rewritten in terms of the conditional mutual information terms as in Eq.~(\ref{eq:expansion}). Then the same argument as in the proof of Theorem 3.5 of Ref.~\cite{CEHHOR} [see Eq.~(57) therein] applies.

Finally, for classical communication, we use the form given in Eq.~\eqref{eq:expansion} to verify the inequality stated below for the case when $A_i$ produces locally the variable $C_i$ and then broadcasts it to all the parties in the form of $C_j$ for $j\neq i$ and to the adversary in the form of $C_{N+1}$ 
(note that broadcasting followed by a partial trace, if needed, can simulate any classical communication among $N$ parties):
\begin{align}
    &I(A_1:\ldots  :A_iC_i:\ldots  :A_N|E)_{\rho}\label{eq.(i)}\\
    &\stackrel{(I)}{=}I(A_iC_i:A_2:\ldots  :A_{i-1}:A_1:A_{i+1}:\ldots  :A_N|E)_{\rho}\\
    &=I(A_iC_i:A_2|E)_{\rho} + I(A_3:A_iC_iA_2|E)_{\rho}+\nonumber\\
    &I(A_4:A_iC_iA_2A_3|E)_{\rho} +\ldots  +I(A_N:A_iC_i
    \ldots  A_{N-1}|E)_{\rho}\label{eq:(ii)}\\
    &\stackrel{(II)}{\ge} I(A_iC_i:A_2C_2|EC_{N+1})_{\rho} \nonumber\\
    &+ I(A_3C_3:A_iC_iA_2C_2|EC_{N+1})_{\rho}\nonumber\\
    &+I(A_4C_4:A_iC_iA_2C_2A_3C_3|EC_{N+1})_{\rho}\nonumber\\
    &+\ldots  +I(A_NC_N:A_iC_i
    \ldots  A_{N-1}C_{N-1}|EC_{N+1})_{\rho}\\
    &=I(A_1C_1:\ldots  :A_iC_i:\ldots  :A_NC_N|EC_{N+1})_{\rho}.
\end{align}
In Eq.~\eqref{eq.(i)} [step $(I)$], we transposed the labels $i$ and $1$ (see Remark \ref{rem:permutations}), and step $(II)$ [Inequality ~\eqref{eq:(ii)}] follows (termwise) from the monotonicity of the (tripartite) mutual information function proved in  Ref.~\cite{CEHHOR} (see the proof of Theorem~$3.5$ therein).

Regarding asymptotic continuity, we consider two
states $\rho_{N(A)}$ and $\sigma_{N(A)}$ such that
$\norm{\rho_{N(A)} - \sigma_{N(A)}}_1\leq \epsilon$. Then,
as in the proof of Theorem $3.5$ of Ref.~\cite{CEHHOR}, for
any map $\Lambda:E\rightarrow E'$ there is
$\norm{\rho_{N(A)}' - \sigma_{N(A)}'}_1\leq \epsilon$,
where $\rho_{N(A)}':= \id_{N(A)} \otimes \Lambda\rho_{N(A)}$
and $\sigma_{N(A)}':= \id_{N(A)} \otimes \Lambda(\sigma_{N(A)})$.
Then, by the expansion Eq.~\eqref{eq:expansion} we obtain
\begin{align}
    &|I(A_i:A_{1}\ldots  A_{i-1}|E')_\rho - I(A_i:A_{1}\ldots  A_{i-1}|E')_\sigma | \nonumber \\&\leq 2\epsilon\log_2 d_{A_i}+ 2g(\epsilon),
\end{align}
with $d_{A_i} := \dim (\mathcal{H}_{A_i})$ and $g(\epsilon):=(1+\epsilon)\log_2(1+\epsilon) - \epsilon \log_2\epsilon$, where we use Lemma~\ref{lem:Shirikov} from Appendix \ref{sec:app:A}, provided in Ref.~\cite{Shi17}.
Hence, in total we get 
\begin{align}
    &|I(A_1:\ldots  :A_N|E')_\rho - I(A_1:\ldots  :A_N|E')_\sigma| \nonumber\\ &\leq
    2\epsilon \sum_{i=1}^{N-1}\log d_{A_i} + (N-1)2g(\epsilon) \nonumber \\ &\leq
    (N-1)[2\epsilon \max_{i\in\{1,\ldots  ,N\}}\log_2 d_{A_i} + 2g(\epsilon)].
\end{align}
For a finite natural $N$, the right-hand side of the above approaches $0$, with $\epsilon \rightarrow 0$.

The subadditivity follows again from the fact
that we can split the term $I(A_1:\ldots  :A_N|E)$ into 
$N-1$ terms of the form $I(A_i:A_1\ldots  A_{i-1}|E)$. Further
treating $A_1\ldots  A_{i-1}$ together as $B_i$ (equivalent of $B$ in the proof of Theorem 3.5 in Ref.~\cite{CEHHOR}), we can
prove the additivity of the form
\begin{equation}
    I(A_iA_i':B_iB_i'|EE')= I(A_i:B_i|E) +I(A_i':B_i'|E')
\end{equation}
for each term and notice the subadditivity from the
fact that in the infimum in the definition of $E_{sq}^{c}(\rho,M)$ there are product channels; hence
in general the formula can be lower than the above.

Finally, we consider normalization. It is straightforward
to see that, assuming $d_{A_i}= d_A$ for each $i \in \{1,\ldots  ,N\}$, on the state representing the ideal key
$\tau_{N(A)E}= \frac{1}{d_{A}}\sum_{i_1=0}^{d_{A}-1} \op{ii\ldots  i}\otimes \tau_{E}$~[Eq. \eqref{eq:key-state}] there is 
$E_{sq}^{c}(\tau,{\mathrm M}) = (N-1)\log_2 d_{A}$, by
noticing that on the product state $I(A_i:A_1\ldots  A_{i-1}|E)=I(A_i:A_1\ldots  A_{i-1}) = \log_2 d_{A}$,
and there are $(N-1)$ of such terms in the definition of $E_{sq}^{c}$. We assume here also that the measurement
${\mathrm M}$ is generating the key in the computational basis.
\end{proof}

We have further an analog of Observation
$4$ of Ref.~\cite{KHD21}. Its proof 
goes along similar lines. Indeed, it does not
depend on either the type of objective function that is minimized or the number of parties; hence we omit it here.
\begin{observation}
For an $N$-partite state $\rho_{N(A)}$
and a POVM ${\mathrm M}_{N(A)}={\mathrm M}_{A_1},\ldots , .{\mathrm M}_{A_N}$ there is
\begin{align}
    &E^{c}_{sq}(\rho,{\mathrm M})\nonumber \\
    &= \inf_{\rho_{N(A)E} = Ext(\rho_{N(A)})} I(A_1:\ldots  :A_N|E)_{{\mathrm M_{N(A)}}\otimes \id_E\rho_{N(A)E}},
\end{align}
where $Ext(\rho_{N(A)})$ stands for the state extension of $\rho_{N(A)}$, i.e., $\rho_{N(A)E}$ is a density operator such that $\Tr[\rho_{N(A)E}]=\rho_{N(A)}$.
\label{obs:extensions}
\end{observation}

Owing to Observation \ref{obs:extensions}, we can obtain
the analog of Lemma $6$ of Ref.~\cite{KHD21}, which
states that $E_{sq}^{c}$ is convex. 
\begin{lemma}
For a tuple of POMVs ${\mathrm M}_{N(A)}$, two
states $\rho_{N(A)}^{(1)}$ and $\rho_{N(A)}^{(2)}$, and
$0<p<1$, there is
\begin{equation}
    E^{c}_{sq}(\bar{\rho}_{N(A)})\leq p E^{c}_{sq}(\rho_{N(A)}^{(1)}) + (1-p)E^{c}_{sq}(\rho_{N(A)}^{(2)}),
\end{equation}
where $\bar{\rho}_{N(A)}=p\rho_{N(A)}^{(1)}+(1-p)\rho_{N(A)}^{(2)}$.
\label{lem:convex}
\end{lemma}
\begin{proof}
The proof is due to the fact that the function
$E^{c}_{sq}(\bar{\rho}_{N(A)})$ is upper bounded
by $I(A_1:\ldots  :A_N|EF)$ evaluated on a state
$\rho_{N(A)EF}={\mathrm M}_{N(A)}\otimes \id_{EF}(p\rho^{(1)}\otimes\op{0}_F+(1-p)\rho^{(2)}\otimes|1\>\<1|_F$. Further, by Eq.~\eqref{eq:expansion} there
is 
\begin{align}
    &I(A_1:\ldots  :A_N|EF)_{\rho} \nonumber\\
    &=pI(A_1:\ldots  :A_N|E)_{{\mathrm M}_{N(A)}\otimes \id_{EF} \rho^{(1)}_{N(A)E}} \nonumber\\
    &+(1-p)I(A_1:\ldots  :A_N|E)_{{\mathrm M}_{N(A)}\otimes \id_{EF} \rho^{(2)}_{N(A)E}}.
\end{align}
Since the states $\rho^{(1)}$ and $\rho^{(2)}$ were arbitrary, we get the thesis.
\end{proof}
We further note that switching from a bipartite key distillation task to the conference key distillation
does not alter the formulation or the proof
of Lemma~$7$ of Ref.~\cite{KHD21}. We state it 
below for the sake of the completeness of the further
proofs.
\begin{lemma}
The independent and identically distributed quantum device-independent key achieved by protocols using (for generating the key) a single tuple of measurements $(\hat{x}_1,\ldots  ,\hat{x}_N)\equiv \hat{\bf{x}}$ applied to ${\cal M}$ of a device $(\rho_{N(A)},{\cal M})$ is upper bounded as
\begin{align}
    &K_{DI,dev}^{iid,\hat{\bf{x}}}(\rho_{N(A)},{\cal M}):=\inf_{\epsilon>0}\limsup_{n \to \infty}\sup_{{\cal P}\in LOPC} \nonumber\\
    &\inf_{{(\sigma_{N(A)},{\cal L})\approx_{\epsilon}(\rho_{N(A)},{\cal M})}} \kappa^{\epsilon,\hat{\bf{x}}}_n(\hat{{\cal P}}({\mathrm L}(\sigma_{N(A)})^{\otimes n}) \nonumber\\
    &\leq  \inf_{{(\sigma_{N(A)},{\cal L})=(\rho_{N(A)},{\cal M})}} K_{DD}({\cal L}(\hat{\bf{x}})\otimes \id_E \psi^\sigma_{N(A)E}),
\end{align}
where ${\mathrm L}\equiv{\cal L}(\hat{\bf{x}})$ is a single pair of measurements induced by inputs $\hat{\bf{x}}$ on $\cal L$ and   $\kappa_n^{\epsilon,\hat{\bf{x}}}$ is the rate of the  $\varepsilon$-perfect conference key achieved and classical labels from local classical operations in $\hat{{\cal P}}\in CLOPC$ are possessed by the allies holding systems $A_i$ for $i \in \{1,\ldots  ,N\}$.
\label{lem:ubound_by_key}
\end{lemma}
Combining Theorem~\ref{thm:bound_on_ddkey} with
Lemma~\ref{lem:ubound_by_key}, we obtain the main 
result of this section. This is a bound
by the reduced reduced c-squashed entanglement. 
\begin{theorem}
The independent and identically distributed quantum device-independent conference key achieved by protocols using a single tuple of measurements $(\hat{x}_1,\ldots  ,\hat{x}_N)\equiv \hat{\bf{x}}$ applied to $ {\cal M}$ of a device $(\rho_{N(A)},{\cal M})$ is upper bounded as  
\begin{align}
    &K_{DI,dev}^{iid,\hat{\bf{x}}}(\rho_{N(A)},{\cal M})\nonumber \\&\leq \frac{1}{N-1}\inf_{{(\sigma_{N(A)},{\cal L})\equiv(\rho_{N(A)},{\cal M})}}E_{sq}^{c}(\sigma_{N(A)},{\cal L}(\hat{\bf{x}}))\\
    &=:E_{sq,dev}^{c}(\rho_{N(A)},{\cal M}(\hat{\bf{x}})).
\end{align}
\label{thm:reduced_dev}
\end{theorem}
We have an analogous result for a key
which is a function only of the tested parameters, that of Bell inequality
violation and the quantum bit error rate.

\begin{theorem}
The independent and identically distributed quantum device-independent key achieved by protocols using (for generating the key) a single tuple of measurements $\hat{\bf x}$ applied to ${\cal M}$ of a device $(\rho_{N(A)},{\cal M})$ is upper bounded as    
\begin{align}
    &K_{DI,par}^{iid,\hat{\bf{x}}}(\rho_{N(A)},{\cal M}) 
    \nonumber\\  
    &\coloneqq \inf_{\epsilon>0}\limsup_{n \to \infty}\sup_{{\cal P}\in LOPC}\nonumber \\&\inf_{\underset{P_{err}(\sigma_{N(A)},{\cal L})\approx_\epsilon P_{err}(\rho_{N(A)},{\cal M})}{\omega(\sigma_{N(A)},{\cal L})\approx_\epsilon\omega(\rho_{N(A)},{\cal M})}} \kappa^{\epsilon,{\hat{\bf x}}}_n({\cal P}({\mathrm L}(\sigma_{N(A)})^{\otimes n}))\label{eq:di_par1}\\
    &\leq  \frac{1}{N-1}\inf_{\underset{P_{err}(\sigma_{N(A)},{\cal L})= P_{err}(\rho_{N(A)},{\cal M})}{\omega(\sigma_{N(A)},{\cal L})=\omega(\rho_{N(A)},{\cal M})}}E_{sq}^{c}(\sigma_{N(A)},{\mathrm L}) \nonumber\\&=: E_{sq,par}^{c}(\rho_{N(A)},{\cal M}(\hat{{\bf x}})),
\end{align}
where $L={\cal L}(\hat{{\bf x}})$ is a single tuple of measurements induced by inputs $\hat{{\bf x}}$ on $\cal L$.
\label{thm:esqbound}
\end{theorem}
For $N=2$, the above bound recovers the result of Ref.~\cite{KHD21}.

In the definition of $E_{sq,par(dev)}^{c}$, one can take the infimum only
over the classical extensions to Eve~\cite{Yang2009}. In that case, for a single input $\hat{\bf x}$ this bound reads $\frac{1}{N-1}I(N(A)\downarrow E)$, as given in Ref.~\cite{Yang2009} (see Ref.~\cite{GW00,CEHHOR,FBL+21} for the bipartite case). We have the following immediate corollary.
\begin{corollary}
The independent and identically distributed quantum device-independent key achieved by protocols using a tuple of measurements $\hat{\bf x}$ applied to a device $(\rho_{N(A)},{\cal M})$ is upper bounded as
\begin{align}
    &K_{DI,dev}^{iid,\hat{{\bf x}}}(\rho_{N(A)},{\cal M}) \nonumber \\&\leq \inf_{(\sigma_{N(A)},{\cal L})=(\rho_{N(A)},{\cal M})} \frac{1}{N-1}I(N(A)\downarrow E)_{P(A_1:\ldots  :A_N|E)} \\
    &\equiv \inf_{(\sigma_{N(A)},{\cal L})=(\rho_{N(A)},{\cal M})} \inf_{\Lambda:~{E\rightarrow F}} \nonumber \\
    &\frac{1}{N-1}I(N(A)| F)_{P(A_1:\ldots  :A_N|\Lambda(E))}, 
    \label{eqn:UBmutual}
\end{align}
where $P(A_1:\ldots  :A_N|E)$ is a distribution coming from measurement ${\cal L}({\hat{\bf x}})$ on purification of $\sigma_{N(A)}$ to system $E$, and the infimum is taken over classical channels transforming a  random variable $E$ to a random variable $F$.
\label{cor:intrinsic}
\end{corollary}
We will exemplify Corollary \ref{cor:intrinsic} for $N=3$ parties and the scenario considered in Ref. \cite{RMW18}. For the results, see Fig.~\ref{fig:attack2}. Let us also note that when one restricts the infimum in Eq.~(\ref{eqn:UBmutual}), the channel $\Lambda :E\rightarrow F$ has only a classical output and the above bound is a multipartite generalization of the {\it intrinsic information} bound given in Ref.~\cite{FBL+21}.
An analogous corollary holds for the case of $K_{DI,par}^{iid,{\hat{\bf x}}}$.

We finally note, that $E_{sq,par}^{c}$ is convex
also, in the multipartite case. This may prove important when one finds upper bounds, as any convexification of two plots obtained from optimization of $E_{sq,par}^{c}$ is then an upper bound on $K_{DI,par}^{c}$, as it was used in Ref.~\cite{KHD21}. We state it below 
following Lemma 8 of Ref.~\cite{KHD21}.
\begin{proposition}
The $E_{sq,par}^{c}$ is convex, i.e.,
for every device $(\bar{\rho},{\cal M})$
and an input tuple ${\bf \hat{x}}$
there is
\begin{align}
    &E_{sq,par}^{c}(\bar{\rho},{\cal M}({\bf \hat{x}})) \leq \nonumber\\ 
    &p_1E_{sq,par}^{c}(\rho_1,{\cal M}({\bf \hat{x}})) +
    p_2 E_{sq,par}^{c}(\rho_2,{\cal M}({\bf \hat{x}})),
\end{align}
where $\bar{\rho}=p_1\rho_1+ p_2 \rho_2$ and $p_1+p_2=1$ with $0\leq p_1 \leq 1$.
\label{prop:convex}
\end{proposition}
\begin{proof}
The proof goes the same way as that for the
bipartite case of Lemma~8 in Ref.~\cite{KHD21}, with the only change that we base it on the convexity of its multipartite version $E_{sq}^{c}$, i.e., Lemma~\ref{lem:convex} here,
and the fact that
\begin{align}
    &I(A_1A_1':\ldots  :A_NA_N'|E)[\rho_{A_1\ldots  A_N}\otimes|i\ldots  i\>\<i\ldots  i|_{A_1'\ldots  A_N'}]\nonumber \\ 
    &= I(A_1:\ldots  :A_N|E)[\rho_{A_1\ldots  A_N}],
\end{align}
where $i \in\{0,1\}$, $\rho_{A_1,\ldots , A_N}$ is arbitrary state of systems $A_1\ldots  A_N$, and we define $I(A_1:\ldots  :A_N|E)[\rho]\equiv I(A_1:\ldots :A_N|E)_\rho$. This is because a pure product state alters neither the entropy of marginals nor the global entropy of the state.
\end{proof}

We note that the multipartite function $E^{c}_{sq}$ can be defined for multiple measurements as in Ref.~\cite{KHD21} and the analogous results (e.g., Corollary 6 of Ref.~\cite{KHD21}) to the bipartite case would hold for the multipartite case.
\begin{definition}
The reduced c-squashed entanglement of the collection of measurements $\mc{M}$ with probability distribution $p({\bf x})$ of the input reads
\begin{equation}
    E^{c}_{\sq}(\rho_{N(A)},\mc{M},p({\bf x}))\coloneqq \sum_{{\bf x}}p({\bf x})E^{c}_{sq}(\rho_{N(A)},\mathrm{M}_{{\bf x}}).
\label{eq:more_measurements}
\end{equation}
\end{definition}
Usually, the parties broadcast their inputs used to generate the key during the protocol.
One can therefore consider a version of the distillable device-independent key achieved by such protocols which do this broadcasting.
We then consider the quantum device-independent key rate
\begin{align}
    & K^{iid, broad}_{DI, dev}(\rho_{N(A)},\mc{M},p({\bf x}))\coloneqq \nonumber\\
    & \inf_{\varepsilon>0}\limsup_{n \to \infty}\sup_{\hat{\mc{P}}\in LOPC}\inf_{(\sigma,\mc{N})\approx_{\varepsilon}(\rho,\mc{M})}  \nonumber\\
    & \kappa^{\varepsilon}_{n}(\hat{\mc{P}}([\sum_{{\bf x}}p({\bf x})\mathrm{N}_{{\bf x}}\otimes\id_E(\op{\psi_{\sigma}}\otimes\op{{\bf x}}_{E_{{\bf x}}})]^{\otimes n})),
\end{align}
where by $broad$ we mean that ${\bf x}\coloneqq (x_1,\ldots , x_N)$ are broadcasted and we make it explicit by adding classical registers $E_{{\bf x}}\coloneqq E_{x_1},\ldots , E_{x_N}$ held by Eve. We have then a generalization of
Theorem of \ref{thm:reduced_dev} to the case of more measurements that are revealed during the protocol of key distillation.

\begin{proposition}
The independent and identically distributed quantum device-independent key achieved by protocols using measurements of a device $(\rho_{N(A)},\mc{M})$ with probability $p({\bf x})$ is upper bounded as 
\begin{align}
  &  K^{iid,broad}_{DI,dev}(\rho_{N(A)},\mc{M},p({\bf x}))\equiv \nonumber\\
& \inf_{\varepsilon>0}\limsup_{n \to \infty}\sup_{\hat{\mc{P}}\in LOPC}\inf_{(\sigma_{N(A)},\mc{N})\approx_{\varepsilon}(\rho_{N(A)},\mc{M})}  \nonumber\\
& \kappa^{\varepsilon}_{n}(\hat{\mc{P}}([\sum_{{\bf x}}p({\bf x})\mathrm{N}_{{\bf x}}\otimes\id_E(\op{\psi_{\sigma}}\otimes\op{{\bf x}}_{E_{{\bf x}}})]^{\otimes n}))\\
&\leq\frac{1}{N-1} \inf_{(\sigma,\mc{N})=(\rho,\mc{M})}E^{c}_{sq}(\sigma_{N(A)},\mc{N},p({\bf x}))\\
& =: E^{c}_{sq, dev}(\rho_{N(A)},\mc{M},p({\bf x})),
\end{align}
where $\mathrm{N}_{{\bf x}}$ are measurements induced by ${\bf x}$ on $\mc{N}$.
\end{proposition}
\begin{proof}
The proof follows straightforwardly from
generalization of Lemma~10 and Theorem~10 from Ref.~\cite{KHD21} for the case of $E_{sq}^{c}$, taking 
as the argument the measurements as in Eq.~\eqref{eq:more_measurements}, 
composed with a broadcast map which for the 
choice of inputs ${\bf x}$ creates systems $E_{\bf x}$
in state $\op{{\bf x}}$.
\end{proof}

\section{Bound on the rate of a Parity CHSH based protocol by the reduced c-squashed entanglement}
\label{sec:with_figure}
In this section we consider the scenario 
of $N=3$ parties and compare the known lower
bound on the conference key rate \cite{RMW18}
with the upper bounds introduced in previous sections. 

Below we exemplify the use 
of the bound by the reduced {c-squashed entanglement} $E_{sq}^{c}$ in the case with classical Eve, that is, when the infimum in its definition runs over the extensions of the form $\sum_i p_i\rho_{A(N)}^i\otimes\op{i}_E$ (or equivalently the channels acting on system $E$ have only classical outputs). We restrict ourselves to the standard protocols with a single pair of inputs generating the key~\cite{FBL+21}. We exemplify the bound given in Corollary \ref{cor:intrinsic} by means of $I(N(A)\downarrow E)$. It then is in essence a matter of  checking the value of the multipartite intrinsic information measure of a distribution which is the output of a key-generating measurement on the attacking state (as it is done in the bipartite case in Ref.~\cite{FBL+21}).

To compare the introduced upper bounds with the known lower bound, for the honest implementation, we focus on the GHZ state, on which depolarizing noise acts locally on three qubits~\cite{RMW18}. 
Having this state, and playing a tripartite game on it \cite{Bell-nonlocality}, called the parity  Clauser-Horne-ShimonyHolt (CHSH) game, one can obtain (in the low-noise regime) a secure conference key.
More precisely, we have the following.
\begin{definition}[parity CHSH game \cite{RMW18}]
The parity CHSH inequality extends the CHSH inequality to N parties as follows. Let Alice and Bob${}_1$, \ldots  , Bob${}_{N-1}$ be the $N$ players of the following game (the parity CHSH game). Alice
and Bob${}_1$ are asked uniformly random binary questions
$x\in \{0,1\}$ and $y\in \{0,1\}$, respectively. The other Bobs are
each asked a fixed question, e.g., always equal to~$1$. Alice
will answer bit $a$, and for all ${i \in \{1,\ldots  ,N-1\}}$, Bob${}_i$ answers bit $b_i$. We denote by $\bar{b} \coloneqq \bigotimes_{2\le i \le N-1} b_i$ the parity of all the answers of Bob${}_{2}$, \ldots  , Bob${}_{N-1}$. The players win if and only if
\begin{align}
    a+b_1=x(y+\bar{b})\mod 2.
\end{align}
\end{definition}
As for the CHSH inequality, the winning probability $P_\mathrm{win}^\mathrm{Parity-CHSH} $ for the classical strategies of the parity CHSH game must satisfy
\begin{align}
    P_\mathrm{win}^\mathrm{Parity-CHSH} \le \frac{3}{4}.
\end{align}
The above inequality can be violated with the  $\Phi^{\GHZ}_{\vv{3}}$ state, with the maximal (quantum) value of $\frac{1}{2}+\frac{1}{2\sqrt{2}}$.

We adopt the same model of noise as in Ref.~\cite{RMW18}, which is represented by qubit depolarizing channels acting the same way on each qubit of the GHZ state:
\begin{align}
    \mathcal{D}_\nu(\rho)= (1-\nu) \rho + \nu \frac{\mathbbm{1}}{2}.
\end{align}
Below we explain the result of applying this global channel to the GHZ state $|\Phi^{\GHZ}_{\vv{N}}\>\<\Phi^{\GHZ}_{\vv{N}}|$ in the case of $N=3$.

\begin{observation} The GHZ state after the action of depolarizing noise on each qubit reads
\begin{align}
    &\mathcal{D}_\nu\otimes \mathbbm{1}_{B_1B_2} (|\Phi^{\GHZ}_{\vv{3}}\>\<\Phi^{\GHZ}_{\vv{3}}|_{AB_1B_2}) \nonumber\\
    &= (1-\nu) |\Phi^{\GHZ}_{\vv{3}}\>\<\Phi^{\GHZ}_{\vv{3}}|_{AB_1B_2} + \nu \frac{\mathbbm{1}_A}{2} \otimes \kappa_{B_1B_2},
\end{align}
 where the $\kappa_{B_1B_2}=\frac{1}{2}\left(|00\>\<00|_{B_1B_2}+|11\>\<11|_{B_1B_2}\right)$ state is separable.
\end{observation}
\begin{remark}
The fully separable state originating from a depolarizing channel (single party), i.e., $\frac{\mathbbm{1}_A}{2} \otimes\kappa_{B_1B_2}$, can not violate the parity CHSH inequality.
\end{remark}
After applications of the depolarizing channel to each of three qubits we obtain the following.
\begin{corollary}
We have
\begin{align}
    &\mathcal{D}_\nu^{\otimes 3} (|\Phi^{\GHZ}_{\vv{3}}\>\<\Phi^{\GHZ}_{\vv{3}}|_{AB_1B_2}) \nonumber\\
    &= (1-\nu)^3 |\Phi^{\GHZ}_{\vv{3}}\>\<\Phi^{\GHZ}_{\vv{3}}|_{AB_1B_2} + [1-(1-\nu)^3] \chi_\nu,
    \label{eq:dep_GHZ_state}
\end{align}
where $\chi_\nu$ is a fully separable state
which reads
\begin{align}
    \chi_\nu:& = \frac{1}{1-(1-\nu)^3
    }\nonumber\\& \times\left((1-\nu)^2\nu\kappa_{AB_1}\otimes \frac{{\mathbbm{1}}_{B_2}}{2} +
    (1-\nu)^2\nu\kappa_{AB_2}\otimes \frac{{\mathbbm{1}}_{B_1}}{2}\right.\nonumber\\
    &\left.+(1-\nu)^2\nu\kappa_{B_1B_2}\otimes \frac{{\mathbbm{1}}_{A}}{2}
        +(3-2\nu)\nu^2\frac{\mathbbm{1}_{AB_1B_2}}{2}\right).
\end{align}
\end{corollary}

In Ref.~\cite{RMW18}, the expected winning probability for the parity CHSH game (with respect to the depolarizing noise parameter) is calculated:
\begin{align}
    p_\mathrm{exp} := \left[\frac 12 + \frac{(1-\nu)^N}{2 \sqrt{2}} +\frac{(1-\nu)^2(1-(1-\nu)^{N-2})}{8\sqrt{2}}\right].
\end{align}
From the above equality for $N=3$, the state in Eq.~(\ref{eq:dep_GHZ_state}) violates the classical bound of $\frac 34$ for $0\leq\nu < \nu_\mathrm{crit}$, where $\nu_\mathrm{crit}\approx 0.1189$.

In this place, we start the construction of the eavesdropper strategy. According to the DI-CKA protocol in Ref.~\cite{RMW18}, the ranges of inputs and outputs are $x \in \{0,1\}$, $y_1 \in \{0,1,2\}$, $y_2 \in \{0,1\}$, and $a,b_1, b_2 \in \{0,1\}$. The setting $(x,y_1,y_2)=(0,2,0)$ associated with measurements of $\sigma_z$ observable is the key-generating round:
\begin{align}
    &P_\nu (a,b_1,b_2|x,y_1,y_2)\nonumber \\
    &=\Tr \left[M_{a|x}\otimes M_{b_1|y_1} \otimes M_{b_2|y_2} \mathcal{D}_\nu^{\otimes 3} (|\Phi^{\GHZ}_{\vv{3}}\>\<\Phi^{\GHZ}_{\vv{3}}|_{AB_1B_2}) \right]\\
    &=(1-\nu)^3\Tr \left[M_{a|x}\otimes M_{b_1|y_1} \otimes M_{b_2|y_2} |\Phi^{\GHZ}_{\vv{3}}\>\<\Phi^{\GHZ}_{\vv{3}}|_{AB_1B_2} \right] \nonumber\\
    &+ (1-(1-\nu)^3) \Tr \left[M_{a|x}\otimes M_{b_1|y_1} \otimes M_{b_2|y_2} \chi_\nu \right]\\
    &=(1-\nu)^3 P_\mathrm{GHZ} (a,b_1,b_2|x,y_1,y_2) \nonumber\\
    &+ (1-(1-\nu)^3) P_\nu^\mathrm{L}(a,b_1,b_2|x,y_1,y_2).
\end{align}
Here the behavior $P_\mathrm{GHZ}$ arises from measurements of the GHZ state (which allows us to violate the classical bound maximally). The local behavior $P_\nu^\mathrm{L}$ arises from the same measurements ($\sigma_z$ observable) for biseparable state and therefore can be expressed as a convex combination of deterministic behaviors.

Eve prepares a convex combination attack~\cite{AGM+06,AMP+06}
\begin{align}
    &P_\nu^\mathrm{CC} (a,b_1,b_2,e|x,y_1,y_2) \nonumber\\
    &= (1-\nu)^3 P_\mathrm{GHZ} (a,b_1,b_2|x,y_1,y_2) \delta_{e,?} \nonumber\\
    &+ [1-(1-\nu)^3] P_\nu^\mathrm{L}(a,b_1,b_2|x,y_1,y_2) \delta_{e,(a,b_1,b_2)}.
\end{align}
This attack might not be optimal as it uses a particular decomposition of $P_\nu$. In order to optimize the attack, Eve should find a decomposition with a maximal weight of local behavior [$1-(1-\nu)^3$ here].



We now consider a particular strategy of post-processing the data which is in Eve's possession, represented by a channel $E\to F$ in Corollary~\ref{cor:intrinsic}. Following Ref.~\cite{FBL+21}, we consider only the distribution coming from a key-generating measurement, which according to the protocol of Ref.~\cite{RMW18} is $X=0$ for Alice and $B_1=2$ and $B_2=0$ for the Bobs in the case of $N = 3$,

\begin{align}
    &P_\nu^\mathrm{ATTACK} (a,b_1,b_2,f|020)=\Lambda_{E\to F} ~P_\nu^\mathrm{CC} (a,b_1,b_2,e|020) \nonumber\\
    &= (1-\nu)^3 P_\mathrm{GHZ} (a,b_1,b_2|020) \delta_{f,?} \nonumber\\
    &+ (1-(1-\nu)^3) P_\nu^\mathrm{L}(a,b_1,b_2|020) \nonumber\\
    &\times \left[\delta_{a,b_1,b_2}\delta_{f,a}+(1-\delta_{a,b_1,b_2})\delta_{f,?} \right],
    \label{eqn:attack1}
\end{align}
where $\delta_{a,b_1,b_2}$ is $1$ if all three indices have the same value and $0$ otherwise. The above attack strategy is therefore a direct three-partite generalization of strategy proposed in Ref.~\cite{FBL+21}. The eavesdropper aims to be correlated only with the events $(a,b_1,b_2)=(0,0,0)$ or $(a,b_1,b_2)=(1,1,1)$, whenever they originate from the local behavior $P_\nu^\mathrm{L}$, and maps all other events to $f=?$. By applying the above attack strategy, we are ready to plot an upper bound on the reduced  c-squashed entanglement shown in Corollary~\ref{cor:intrinsic}.
The latter bound is a multipartite version of the intrinsic information \cite{MauWol97c-intr,Intrinsic-Maurer}, used first for the bipartite case in \cite{AMP+06} against non-signaling adversary (see in this context \cite{KW17,PKBW21,WDH22}). Here the strategy of Eve to process her classical variable $E$ to $F$ is based on \cite{FBL+21} as shown above.

\section{Gap between DI-CKA and DD-CKA}
\label{sec:gap}
In this section, we provide a bound on the conference key agreement of $N$ parties in terms of the bounds for groupings of these parties into groups of fewer than $N$ users.
We further show that 
there is a gap between the
device-independent and device-dependent conference key agreement rates. This gap implies that there are states for which there are no measurements used for testing and no CLOPC protocol that can achieve the same number of keys as in the device-dependent case. The gap is inherited from the analogous gap shown for the bipartite case~\cite{CFH21}.

In what follows, by a (nontrivial) partition {\cal P} of 
the set of systems $\{A_1,\ldots  ,A_N\}$,
we mean any grouping of the systems 
into at least two but no more than $N-1$ subsets such that each $A_i$ belongs to exactly one subset and each of them belongs to some subset. 

Let us now generalize the definition of the reduced device-dependent key to the case of the conference key agreement. We will further also show the fact that the latter quantity bounds the device-independent conference key (i.e.,~Theorem $6$ of Ref.~\cite{CFH21}).
\begin{definition}
The reduced device-dependent conference key rate of an $N$-partite state $\rho_{N(A)}$ reads
\begin{align}
K^{\downarrow}(\rho_{N(A)})\coloneqq \sup_{\cal M}\inf_{(\sigma_{N(A)},{\cal L})=(\rho_{N(A)},{\cal M})} K_{DD}(\sigma_{N(A))}).
\end{align}
\end{definition}
A direct analog of Theorem $6$ of Ref.~\cite{CFH21} (with an analogous proof which we omit here) states that the {\it reduced } device-dependent key upper bounds the device independent key.
\begin{theorem}
For any $N$-partite state $\rho_{N(A)}$ and any ${\cal M}$, there is
\begin{equation}
    K_{DI}(\rho_{N(A)},{\cal M}) \leq \inf_{(\sigma_{N(A)},{\cal L})=(\rho_{N(A)},{\cal M})}K_{DD}(\sigma_{N(A)})
\end{equation}
and in particular,
\begin{equation}
    K_{DI}(\rho_{N(A)})\equiv \sup_{\cal M} K_{DI}(\rho_{N(A)},{\cal M}) \leq K^{\downarrow}(\rho_{N(A)}).
\end{equation}
\label{thm:downarrow}
\end{theorem}

We first observe the following bound.
\begin{proposition} For any $N$-partite quantum 
behavior $(\rho_{N(A)},{\cal M})$ there is
\begin{align}
&K_{DI,dev}^{iid}(\rho_{N(A)},{\cal M})\leq \min\left\{ \min_{\cal P} K_{DI,dev}^{iid}(\rho_{{\cal P}(N(A))}),\right.\nonumber\\& \left.\min_{{\cal P}} \inf_{(\sigma_{{\cal P}(N(A))},{\cal L})=(\rho_{{\cal P}(M(A))},{\cal M})}K_{DD}(\sigma_{{\cal P}(N(A))})\right\},
\label{eq:trivial_bound}
\end{align}
where ${\cal P}$ is any non-trivial partition of the set of systems $A_1,\ldots , A_N$.
\label{prop:cuts}
\end{proposition}
\begin{proof}
The proof of the bound by $K_{DI,dev}^{iid}(\rho_{{\cal P}(N(A))})$ follows from the fact that
any protocol of distillation of the
DI conference key from the $N$-partite state is a special case
of a protocol that distills the 
DI conference key from a non-trivial partition ${\cal P}$. This is because the class of LOPC protocols in these two scenarios is
in relation to $LOPC(A_1,\ldots  ,A_N)\subsetneq LOPC({\cal P}(A_1,\ldots  ,A_N))$.
The other bound follows from the fact that 
for any grouping ${\cal P}$, by
Theorem \ref{thm:downarrow} above,

\begin{align}
    &K_{DI,dev}^{iid}(\rho_{{\cal P}(N(A))},{\cal M}) \nonumber\\&\leq   \inf_{(\sigma_{{\cal P}(N(A))},{\cal L})=(\rho_{{\cal P}(M(A))},{\cal M})}K_{DD}(\sigma_{{\cal P}(N(A))}).
\end{align}
\end{proof}

An analogous fact to the above holds for
$K_{DI,par}^{iid}$ as well.

Following Ref.~\cite{CFH21}, we show now that there is a gap between the numbers of conference keys
and device-independent conference keys. We will use 
the fact that there it has been proven that there are states with $K^{\downarrow}(\rho_{AB})< K_{DD}(\rho_{AB})$. From such state $\rho_{AB}$ we construct a multipartite state with the property that $K_{DI}(\rho_{N(A)}) < K_{DD}(\rho_{N(A)})$, as it is described in the proof of the following theorem.

\begin{theorem}
Let 
$\rho_{AB}\in\mc{B}(\mc{H}_A\otimes\mc{B})$, where $\dim(\mc{H}_A)=d_A$ and $\dim(\mc{H}_B)=d_B$, be a bipartite state which 
admits a gap $K_{DD}(\rho_{AB}) - K^{\downarrow}(\rho_{AB})\geq c >0$ for some constant $c$. Then for any $N$ there
is a multipartite state
$\rho_{N(A)}$ with
local dimensions at most $d_A\times d_B$ with
$K_{DD}(\rho_{N(A)})-K_{DI}(\rho_{N(A)}) \geq c$.
\label{thm:gap}
\end{theorem}
\begin{proof}
Consider a state $\rho_{N(A)}$ constructed as a path
made of state $\rho_{AB}$ (as, e.g., in a line of a quantum repeater):
\begin{align}
    \widetilde{\rho}_{N(A)}\coloneqq 
    \rho_{A_1^1A_1^2}\otimes \rho_{A_2^1A_2^2}\otimes\rho_{A_3^1A_3^2}\otimes\ldots  \otimes\rho_{A_{N-1}^1A_{N-1}^2}.
    \label{eq:path}
\end{align}
Here $\rho_{A_1^1A_1^2}=\rho_{A_2^1A_2^2}= \ldots =\rho_{A_{N-1}^1A_{N-1}^2}=\rho_{AB}$ and by the way of notation we have
$A_1^1\equiv A_1$ and $A_1^2A_2^1\equiv A_2$,\ldots  , $A_{N-1}^2\equiv A_N$.
That is, the first party has only system $A_1^1$ and the last only system $A_{N-1}^2$,
while the $i$th party for $1<i<N$ has systems $A_{i-1}^2A_i^1$ at hand.

Since the states $\rho_{A_i^1A_i^2}$ form
a {\it spanning tree}  of a graph of $N$ systems (in fact a path),
we can follow the lower bound given in Section~VI~A of Ref.~\cite{DBWH19} and note that \begin{equation}
    K_{DD}(\widetilde{\rho}_{N(A)})\geq \min_i K_{DD}(\rho_{A_i^1A_i^2})=K_{DD}(\rho_{AB}).
    \label{eq:dd_bound}
\end{equation}
Indeed, the parties can first distill a key at rate $K_{DD}(\rho_{AB})$ along the edges of the path. Denote such distilled keys by $k_{ij}$ between nodes $i$ and $j$. Further, $A_1$ can XOR her key $k_{12}$ with a locally generated private random bit string $r$ of length $K_{DD}(\rho_{AB})$ and send $k_{12}\oplus r$ to $A_2$;
further, $A_2$ can obtain $r=k_{12}\oplus (k_{12}\oplus r)$ and send it to the next party by XORing it with the key $k_{23}$. This process repeated $N-1$ times, leaves all the parties knowing $r$, which remained secret due to one-time pad encryption by the keys $k_{12},k_{23},\ldots  ,k_{N-1,N}$.
It then suffices to note that,
by Proposition~\ref{prop:cuts},
\begin{align}
&K_{DI}(\widetilde{\rho}_{N(A)},{\cal M})\leq 
\nonumber \\&\inf_{(\sigma_{A^1_1:(A_1^2,A_2^1\ldots  A_N)},{\cal L})=(\rho_{A^1_1:(A_1^2,A_2^1\ldots  A_N)},{\cal M})} K_{DD}(\sigma_{A_1^1:(A_1^2\ldots  A_N)})  \\&\leq
\inf_{(\sigma_{A^1_1:A_1^2},{\cal L})=(\rho_{A^1_1:A_1^2},{\cal M})} K_{DD}(\sigma_{A_1^1:A_1^2}).
\label{eq:cut_lb}
\end{align} 
This is because there are no more conference keys than the number of device-dependent keys distilled in the cut $A_1^1:(A_1^2,\ldots , A_{N-1}^2)$.
The latter is also upper bounded by the key distilled between systems $A_1^1$ and $A_1^2$. This is due to the fact that any distillation protocol between $A_1^1$ and $(A_1^2,\ldots , A_N)$ is a particular protocol distilling key between systems $A_1^1$ and $A_1^2$.

Taking the supremum over ${\cal M}$ on both sides of the inequality~\eqref{eq:cut_lb}, we obtain
\begin{align}
    &K_{DI}(\widetilde{\rho}_{N(A)})\equiv \sup_{{\cal M}} K_{DI}(\widetilde{\rho}_{N(A)},{\cal M}) \nonumber \\
   \leq & K^{\downarrow}(\rho_{A_1^1:A_1^2})\equiv 
    \sup_{\cal M}\inf_{(\sigma_{A^1_1:A_1^2},{\cal L})=(\rho_{A^1_1:A_1^2},{\cal M})} K_{DD}(\sigma_{A^1_1:A_1^2})\nonumber \\
    & 
    =K^{\downarrow}(\rho_{AB}).
\end{align}
Hence we get $K_{DI}(\rho_{N(A)})\leq K^{\downarrow}(\rho_{AB})$. This fact, by Eq.~\eqref{eq:dd_bound}, and the fact that by assumption $K_{DD}(\rho_{AB}) - K^{\downarrow}(\rho_{AB})\geq c> 0$ imply the following
chain of inequalities:
\begin{equation}
    K_{DD}(\widetilde{\rho}_{N(A)})\geq K_{DD}(\rho_{AB}) > K^{\downarrow}(\rho_{AB}) \geq K_{DI}(\widetilde{\rho}_{N(A)}).
\end{equation} The above implies then the desired gap $K_{DD}(\widetilde{\rho}_{N(A)})-K_{DI}(\widetilde{\rho}_{N(A)})>0$. Moreover, this gap is as large as $c>0$ due to the assumption that 
$K_{DD}(\rho_{AB})-K_{DI}(\rho_{AB})\geq c >0$. The claim about dimensions follows from the form of the state given in Eq.~\eqref{eq:path}.
\end{proof}
From Ref.~\cite{CFH21} we have the immediate corollary that there is a gap between the DI-CKA and DD-CKA.
\begin{corollary} For any $N$ there is a state $\widetilde{\rho}_{N(A)}$ for which there is
\begin{equation}
K^{iid}_{DI,dev}(\widetilde{\rho}_{N(A)}) < K_{DD}(\widetilde{\rho}_{N(A)}).
\end{equation}
\end{corollary}
\begin{proof}
Reference~\cite{CFH21} shows an example of a bipartite state $\rho_{AB}$ with the gap $K^{\downarrow}(\rho_{AB})<K_{DD}(\rho_{AB})$. The construction given in Eq.~\eqref{eq:path} based on this $\rho_{AB}$ proves the thesis via Theorem~\ref{thm:gap}.
\end{proof}
We note also that
a bound similar to that in the above corollary holds for $K^{iid}_{DI,par}$ and $K_{DI}$ itself due to the fact that $K^{iid}_{DI}\geq K_{DI}$ by definition \cite{CFH21}.
We can modify the proof technique shown above to see
the following general remark.
\begin{remark}
In the above construction one need not use only the state $\rho_{A_1^1A_1^2}^{\otimes k}$ having  $K^\downarrow(\rho_{A_1^1A_1^2}) < K_{DD}(\rho_{A_1^1A_1^2})$. 
In fact the state on systems $A_2^1A_2^2\ldots  A_{N-1}^1A_{N-1}^2$ can be an {\it arbitrary} state having $K_{DD}(\rho_{A_2^1A_2^2\ldots  A_{N-1}^1A_{N-1}^2})\geq K_{DD}(\rho_{A_1^1A_1^2})$. It can be even a $\Phi^{{\rm GHZ}}_{\vv{N-1}}$ state of arbitrary large local
dimension. This is with no change in the above proof if only $\rho_{A_1^1A_1^2}$ is on systems $A_1^1A_1^2$ with the gap we have mentioned. See Fig.~\ref{fig:multi} for the tripartite example.
\end{remark}

\section{DI-CKA versus genuine nonlocality and entanglement}\label{sec:GNLandENT}
We now discuss the topic of genuine nonlocality and entanglement in the context of the DI-CKA, introducing the notion of {\it quantum locality}.

We say that a behavior $P(\bf{a}|\bf{x})$ is
{\it local in a cut} $(A_{i_1}\ldots A_{i_k}):(A_{i_{k+1}}\ldots  A_{i_N})$ if it can be written as a product
of two behaviors on systems $A_{i_1}..A_{i_k}$ and $A_{i_{k+1}}\ldots  A_{i_N}$, respectively.
The behavior $P(\bf{a}|\bf{x})$ is {\it  genuinely non-local} if and only if it is not 
a mixture of behaviors that are a
product in at least one cut. 

We show that any behavior from
which the parties draw the conference key in a {\it single-shot} (single run) must exhibit genuine nonlocality. The scenario
of a single run was considered in the
context of a non-signaling adversary
\cite{hanggi-2009,masanes-2009-102}. 
For that reason, we depart from the traditional definition of DI quantum key distillation rate by considering
a single-shot DI quantum key distillation rate obtained by
an LOPC post-processing of a distribution
obtained from some behavior $P({\bf a}|{\bf x})$ when all the parties measure all the inputs $\bf x$ in parallel at the same time.

For the purpose of Theorem \ref{thm:single-shot} below,
by a ``local'' set we will mean 
the set of behaviors that are convex 
mixtures of behaviors that are a 
product in some cut and both behaviors
in the product have quantum realization.
We will denote this set by $\mathrm{LQ}$ (locally quantum). Any distribution which is not in $\mathrm{LQ}$ can be treated as genuinely non-local, although other definitions are adopted in the literature \cite{Bell-nonlocality}. Exemplary extreme behavior in this set is a product of the Tsirelson behavior
$P(a_1,a_2|x_1,x_2) =\Tr [\op{\Phi_{\vv{2}}}M^{x_1}_{a_1}\otimes M^{x_2}_{a_2}]$
with $M^{0}_{a_1}=\sigma_z$, $M^1_{a_1}=\sigma_x$, $M^0_{a_2}= \frac{(\sigma_x+\sigma_z)} {\sqrt{2}}$, $M^1_{a_2}=\frac{(\sigma_z-\sigma_x)}{ \sqrt{2}}$ ($\sigma_x$ and $\sigma_z$ being Pauli-$X$ and Pauli-$Z$ operators, respectively), and any deterministic local behavior $P(a_3|x_3)$: $P(a_1,a_2|x_1,x_2)P(a_3|x_3)$.
The theorem which we show below applies to the scenario where all the parties share a {\it single} copy of a device and measure the inputs in parallel to ensure nonsignaling. We give the definition of the key rate obtained in this setup in full analogy to Definition \ref{def:di-dev} as follows.
\begin{definition}
The maximum single-shot device-independent quantum key distillation rate of a device $(\rho,\mc{M})$ with independent and identically distributed behavior is defined as
\begin{equation}
    K^\mathrm{single-shot}_{DI,dev}(\rho,\mc{M})\coloneqq \inf_{\varepsilon>0} \sup_{\hat{\mc{P}}} \inf_{\eqref{eq:k-1}} \kappa^\varepsilon_{n} \left(\hat{\mc{P}}(\sigma,\mc{N})\right),
\end{equation}
where $\kappa_n^\varepsilon$ is the quantum key rate achieved for any security parameter $\varepsilon$ and measurements $\mc{N}$.

Here $\hat{\mc{P}}$ is a protocol composed of classical local operations and public (classical) communication  acting on a single copy of $(\sigma,\mc{N})$ which, composed with the measurement, results in a quantum local operations and public (classical) communication  protocol.
\end{definition}
We are ready to state the following theorem.
\begin{theorem}\label{thm:single-shot}
If a behavior $(\rho_{N(A)},\mc{M})$ satisfies
$K^{single-shot}_{DI,dev}(\rho_{N(A)},\mc{M})>0$ then it is not in $\mathrm{LQ}$.
\end{theorem}
\begin{proof}
The proof goes by contradiction.
Suppose a behavior $p(\bf{a}|\bf{x}) \equiv (\rho_{N(A)},\mc{M})$
is not genuinely nonlocal. That
is, it can be expressed as a convex
combination of behaviors which
are a product in at least one cut denoted by $(A^{(i)}_{j_1}\ldots  A^{(i)}_{j_k}):(A^{(i)}_{j_{k+1}}\ldots  A_{j_N}^{(i)})$ for the $i$th behavior in the combination. We express this as 
\begin{equation}
    \sum_iq_i p_i({\bf a}|{\bf x})_{(A^{(i)}_{j_1}\ldots  A^{(i)}_{j_k}):(A^{(i)}_{j_{k+1}}\ldots  A_{j_N}^{(i)})},
    \label{eq:mixture}
\end{equation}
where $p_i({\bf a}|{\bf x})$ are some quantum behaviors. Consider then a device 
$p_i$ as a bipartite one, with
parties $(A^{(i)}_{j_1}\ldots  A^{(i)}_{j_k})$ together forming $A'$ and $(A^{(i)}_{j_{k+1}}\ldots  A_{j_N}^{(i)})$ forming $A''$.  Such a device
has zero bipartite DI quantum keys, as
it is a product in cut $A':A''$. By virtue of purification, Eve can have
access to the mixture~\eqref{eq:mixture}, knowing
which of the mixing terms $i$ happened.
By Proposition \ref{prop:cuts} we have that
from any of such terms, one can not
draw a conference key, as Eve has
a local hidden variable model for it.
Indeed, the right-hand side of \eqref{eq:trivial_bound} is then $0$, as Eve can adopt an attack which, e.g., makes zero reduced c-squashed entanglements \cite{KHD21}. We thus obtained the desired contradiction.
\end{proof}

Let us now recall the notion of genuine entanglement.
We say that a multipartite state $\rho_{A_1A_2\ldots  A_N}$ is separable in a cut $(A_{i_1}..A_{i_k}):(A_{i_{k+1}}\ldots  A_{i_N})$ if it can be written as convex mixtures of product states between systems $A_{i_1}..A_{i_k}$ and $A_{i_{k+1}}\ldots  A_{i_N}$. If a multipartite state $\rho_{A_1A_2\ldots  A_N}$ can be written as a mixture of separable states that are a product in at least one cut then it is called biseparable. We say that $\rho_{A_1A_2\ldots  A_N}$ is {\it genuinely entangled} if and only if it is not a mixture of separable states that are a product in at least one cut. It was shown in Ref.~\cite{BFF+16} that there exist $N$-partite states for all $N>2$ where some genuinely entangled states admit a fully LRHV model, i.e., where all parties are separated.

Let ${\rm GE}(N(A))$, ${\rm BS}(N(A))$, and ${\rm FS}(N(A))$ denote the set of all $N$-partite states $\rho_{A_1A_2\ldots  A_N}$ that are genuinely entangled, biseparable, and fully separable, respectively (see Ref.~\cite{DBWH19}). For $n$ copies of $N$-partite state, when we consider partition across designated $N$ parties, we denote local groupings by $N(A^{\otimes n})$.

\begin{remark}
It is necessary to consider a single-shot DI key in Theorem \ref{thm:single-shot} because the set of LQ behaviors is not closed under tensor product. This is for the same reason that the set of biseparable states is not closed under a tensor product.
\end{remark}

The following theorem follows from Corollary~\ref{obs:ent-mes} as well as Proposition~2 of Ref.~\cite{DBWH19}.
\begin{theorem}
The maximum device-independent conference key agreement rates of a device $(\rho_{N(A)},\mc{M})$ are upper bounded by
\begin{align}
     K_{DI,dev}(\rho_{N(A)}, \mathcal{M})   & \leq  \inf_{(\sigma_{N(A)},\mc{L})= (\rho_{N(A)},\mc{M})} E^{\infty}_{GE} (\sigma_{N(A)}),\\
       K_{DI,par}(\rho_{N(A)}, \mathcal{M})   & \leq  \inf_{\underset{P_{err}(\sigma_{N(A)},{\cal L})= P_{err}(\rho_{N(A)},{\cal M})}{\omega(\sigma_{N(A)},{\cal L})=\omega(\rho_{N(A)},{\cal M})}} E^{\infty}_{GE}(\sigma_{N(A)}),
     \end{align}
     where $E^{\infty}_{GE} (\varsigma)$ is the regularized relative entropy of genuine entanglement~\cite{DBWH19} of a state $\varsigma_{A_1A_2\ldots  A_N}$, with
     \begin{equation}
         E^{\infty}_{GE} (\varsigma)=\inf_{\varphi\in{\rm BS(N(A^{\otimes n}))}}\lim_{n\to\infty}\frac{1}{n}D(\varsigma^{\otimes n}\Vert \varphi)
     \end{equation}
      where $D(\rho\Vert\sigma)$ is the relative entropy between two states $\rho$ and $\sigma$, with $D(\rho\Vert\sigma)= \Tr\left[\rho(\log_2\rho-\log_2\sigma)\right]$ if $\supp{\rho}\subseteq\supp{\sigma}$; otherwise it is $\infty$~\cite{Ume62}. 
      \label{thm:E_R}
\end{theorem}
 
We note here that there is a trivial bound that can
be obtained from Theorem \ref{thm:E_R} above, which
is encapsulated in the following corollary.
\begin{corollary}
For any state $\rho_{N(A)}$ with $\min_{i \in \{1\ldots  N\}} d_{A_i} =: d$ there is
\begin{align}
    &K_{DI}(\rho_{N(A)})\equiv \sup_{\cal M} K_{DI}(\rho_{N(A)},{\cal M})\leq \nonumber \\&
    \min \{ p \log_2 d: p \in [0,1], \rho = p\rho'+(1-p)\rho_{fs}, \nonumber \\ 
    &\rho_{fs}\in {\rm FS(N(A))}\},
\end{align}
\label{cor:trivial}
\end{corollary}
\begin{proof}
Given any decomposition of a state $\rho_{N(A)}$ into
$\rho_{N(A)} = p \rho' + (1-p) \rho_{fs}$, where the state $\rho_{fs}$ is a fully separable state, we have
\begin{align}
&K_{DI}(\rho_{N(A)})\leq \sup_{\cal M} K_{DI,dev}^{iid}(\rho_{N(A)},{\cal M}) \nonumber\\& \leq \sup_{\cal M} \inf_{(\sigma_{N(A)},{\cal L})=(\rho,{\cal M})}E^{\infty}_{GE}(\sigma_{N(A)})\leq \sup_{\cal M} E_R(\rho_{N(A)}) \nonumber\\& \leq p E_R(\rho')\leq p \min_i \log_2 d_{A_i},\end{align} 
where we have used Theorem \ref{thm:E_R} (also see Corollary~6 of \cite{DBWH19}) and
the fact that $E_R(\rho)=\inf_{\kappa \in {\rm FS}}D(\rho\Vert\kappa)$ is (i) convex, (ii) zero on fully separable states, and (iii) does not exceed the minimum logarithm of dimensions of the input state, which can be proved by noticing
that $E_R(\rho)\leq D(\rho\Vert\rho_{A_i}\otimes \rho_{A_{\neq i}})= I(A_i:A_{\neq i})\leq \log_2 d$ where $A_i$ has minimal dimension among systems $A_1,\ldots , A_N$.
\end{proof}
We presented this bound in Fig.~\ref{fig:attack2} in Section~\ref{sec:with_figure}
and we saw that it is indeed above the upper bounds which we derive in Section~\ref{sec:multi_squashed}.

\section{Conclusion}\label{sec:discussion}
We have demonstrated a number of upper bounds on the quantum secure conference key, generalizing (i) the results of Ref.~\cite{KHD21} regarding a relative entropy based bound and
(ii) the results of Ref.~\cite{FBL+21} regarding the reduced c-squashed entanglement.

Interestingly, the approach of Ref.~\cite{FBL+21} does not result in zero keys in any noise regimes for
the parity CHSH game of Ref.~\cite{RMW18}. It would be important to see if this can be improved by changing Eve's strategy or the bound needs to be changed.

We have also shown that the fundamental gap between device-independent and device-dependent keys also holds in the multipartite case. We have given an exemplary state which is based directly on the bipartite states given in Ref.~\cite{CFH21}. It is interesting if
such a state exists in lower dimensions
or even possibly on $N$ qubits.

Finally, our results hold for the static case of quantum states. The next step would be to generalize the results of Ref.~\cite{KHD21} for the dynamic case of quantum channels to the multipartite scenario.

\medskip
{\it Note added}.--- 
The topic of upper bounds on the DI-CKA is also studied in the parallel work of~\cite{PKBW21}. Comparison between basic approaches (i.e., for the DI quantum key distribution between two honest parties) used in Ref.~\cite{PKBW21} and in this paper to get upper bounds on DI-CKA is discussed in Ref.~\cite{KHD21}.

\begin{acknowledgements}
KH acknowledges the Fulbright Program, Mark Wilde and Cornell ECE for hospitality during the Fulbright scholarship at the School of Electrical and Computer Engineering of Cornell University.
We acknowledge partial support by the Foundation for Polish Science (IRAP Project ICTQT, Contract No. MAB/2018/5, cofinanced by EU within Smart Growth Operational Programme). The International Centre for Theory of Quantum Technologies project (Contract No. MAB/2018/5) is carried out within the International Research Agendas Programme of the Foundation for Polish Science cofinanced by the European Union from the funds of the Smart Growth Operational Programme, axis IV: Increasing the research potential (Measure 4.3). M.W. acknowledges grant Sonata Bis 5 (Grant No. 2015/18/E/ST2/00327) from the National Science Center.  S.D. acknowledges Individual Fellowships at Universit\' {e} libre de Bruxelles; this project received funding from the European Union's Horizon 2020 research and innovation program under the Marie Skłodowska-Curie Grant Agreement No. 801505. S.D. also thanks Harish-Chandra Research Institute, Prayagraj (Allahabad, India) for  hospitality during his visit where part of this work was done. S.D. thanks Maksim E. Shirokov for pointing out to the Theorem 7 of Ref.~\cite{DSW18} and suggesting equality of the reduced c-squashed entanglement $E^c_{sq}(\rho,{\rm M})$ and the dual c-squashed entanglement $\widetilde{E}^c_{sq}(\rho,{\rm M})$ of a device $(\rho,{\rm M})$.

\end{acknowledgements}

\appendix
\section{Continuity statements}\label{sec:app:A}

There are the following lemmas.
\begin{lemma}[Alicki-Fannes-Winter continuity bounds~\cite{Win16}]
For states $\rho_{AB}$ and $\sigma_{AB}$, if $\frac{1}{2}\norm{\rho-\sigma}_1\leq \varepsilon\leq 1$, then 
\begin{equation}
    \abs{S(A|B)_{\rho}-S(A|B)_{\sigma}}\leq 2\varepsilon \log_2 d+g(\varepsilon),
\end{equation}
where $d=\dim(\mc{H}_A)<\infty$ and $g(\varepsilon)\coloneqq (1+\varepsilon)\log_2 (1+\varepsilon)-\varepsilon\log_2 \varepsilon$.
\end{lemma}

\begin{lemma}[from \cite{Shi17}]
If $d=\min\{\dim(\mc{H}_A),\dim(\mc{H}_B)\}<+\infty$, then
\begin{equation}\label{eqn:Shirokov}
    \abs{I(A;B|C)_{\rho}-I(A;B|C)_{\sigma}}\leq 2\varepsilon\log_2 d+2g(\varepsilon)
\end{equation}
for any states $\rho_{ABC}$ and $\sigma_{ABC}$, where $\varepsilon=\frac{1}{2}\norm{\rho-\sigma}_1$.
\label{lem:Shirikov}
\end{lemma}

\section{Secrecy monotones}\label{sec:app:B}
In this Appendix we revisit Theorem~3.1 of Ref.~\cite{CEHHOR} and generalize the result by relaxing the constraints on the Hilbert spaces in the following way. First, we prove an analogy to Lemma $A.1$ of Ref.~\cite{CEHHOR}.

\begin{lemma}[cf. \cite{CEHHOR}]
The maximization in the definition of $K_{DD}$~\eqref{eq:kdd} can be restricted to protocols that use communication at most linear in the number of copies of $\rho_{ABE}$. The eavesdropper system is not necessarily restricted to a finite dimension. 
\label{lem:CEHHOR}
\end{lemma}

\begin{proof}
The proof of Lemma \ref{lem:CEHHOR} goes along the lines of the proof of Lemma $A.1$ in Ref.~\cite{CEHHOR}. The change that is necessary to allow the eavesdropper to hold the system of infinite dimension is the use of asymptotic continuity of the conditional mutual information of  Ref.~\cite{Shi17} (see Lemma~\ref{lem:Shirikov} herein) instead of the Alicki--Fannes inequality. This results in:
\begin{align}
    I(A:B)_\sigma-I(A:E)_\sigma \ge l_{n_0} (1-4\epsilon) -4g(\epsilon),
\end{align}
where $l_{n_0}$ is the length of the output of a distillation protocol using $n_0$ copies of the input state. The state $\sigma$ is the output of the latter protocol. The overall key rate of the modified protocol which has linear communication admits then a lower bound
\begin{align}
    \widetilde{R} \ge (1-4\epsilon)(R-\epsilon)-\frac{4g(\epsilon)}{n_0}.
\end{align}
The other parts of the proof are not altered.
\end{proof}

\begin{lemma}[cf.~\cite{CEHHOR}] \label{lem:cehhor}
Let $E(\rho)$ be a function mapping a tripartite quantum state $\rho_{ABE}$ into positive numbers such that the following hold:
(a) monotonicity, i.e., $E(\Lambda(\rho))\leq E(\rho)$ for any LOPC $\Lambda$; 
(b) asymptotic continuity, i.e., for any states $\rho^{n}$ and $\sigma^{n}$ on $\mc{H}_{A}\otimes\mc{H}_B\otimes\mc{H}_E$, the condition $\norm{\rho^n-\sigma^n}_1\to 0$ implies $\frac{1}{\log_2  r_n}\abs{E(\rho^n)-E(\sigma^n)}\to 0$ where $r_n=\dim(\mc{H}^n_A)$; and
(c) normalization, i.e., $E(\tau^{(l)})=l$.\\
Then the regularization of the function $E$ given by $E^{\infty}(\rho)=\limsup_{n\to\infty}\frac{M(\rho^{\otimes n})}{n}$ is an upper bound on the device-dependent key distillation rate $K_{DD}$, i.e., $E^{\infty}(\rho_{ABE})\geq K_{DD}(\rho_{ABE})$ for all $\rho_{ABE}$ with $\dim_{A}<\infty$, if in addition $E$ satisfies 
(d) subadditivity on tensor products: $E(\rho^{\otimes n})\leq n E(\rho)$; then $E$ is an upper bound on $K_{DD}$. 

\label{lem:infinite_dim}
\end{lemma}
\begin{proof}
The proof arguments are same as those stated in Ref.~\cite{CEHHOR} with relaxation on the Hilbert space of $E$. We observe that the proof arguments hold even when there is no restriction on the $\dim(\mc{H}_E)$, i.e., $E$ can be finite dimensional or infinite dimensional. It suffices to have $\dim(\mc{H}_A)$ be finite dimensional.
\end{proof}

\bibliography{revision}{}

\begin{thebibliography}{53}%
\makeatletter
\providecommand \@ifxundefined [1]{%
 \@ifx{#1\undefined}
}%
\providecommand \@ifnum [1]{%
 \ifnum #1\expandafter \@firstoftwo
 \else \expandafter \@secondoftwo
 \fi
}%
\providecommand \@ifx [1]{%
 \ifx #1\expandafter \@firstoftwo
 \else \expandafter \@secondoftwo
 \fi
}%
\providecommand \natexlab [1]{#1}%
\providecommand \enquote  [1]{``#1''}%
\providecommand \bibnamefont  [1]{#1}%
\providecommand \bibfnamefont [1]{#1}%
\providecommand \citenamefont [1]{#1}%
\providecommand \href@noop [0]{\@secondoftwo}%
\providecommand \href [0]{\begingroup \@sanitize@url \@href}%
\providecommand \@href[1]{\@@startlink{#1}\@@href}%
\providecommand \@@href[1]{\endgroup#1\@@endlink}%
\providecommand \@sanitize@url [0]{\catcode `\\12\catcode `\$12\catcode
  `\&12\catcode `\#12\catcode `\^12\catcode `\_12\catcode `\%12\relax}%
\providecommand \@@startlink[1]{}%
\providecommand \@@endlink[0]{}%
\providecommand \url  [0]{\begingroup\@sanitize@url \@url }%
\providecommand \@url [1]{\endgroup\@href {#1}{\urlprefix }}%
\providecommand \urlprefix  [0]{URL }%
\providecommand \Eprint [0]{\href }%
\providecommand \doibase [0]{http://dx.doi.org/}%
\providecommand \selectlanguage [0]{\@gobble}%
\providecommand \bibinfo  [0]{\@secondoftwo}%
\providecommand \bibfield  [0]{\@secondoftwo}%
\providecommand \translation [1]{[#1]}%
\providecommand \BibitemOpen [0]{}%
\providecommand \bibitemStop [0]{}%
\providecommand \bibitemNoStop [0]{.\EOS\space}%
\providecommand \EOS [0]{\spacefactor3000\relax}%
\providecommand \BibitemShut  [1]{\csname bibitem#1\endcsname}%
\let\auto@bib@innerbib\@empty
\bibitem [{\citenamefont {Dowling}\ and\ \citenamefont {Milburn}(2003)}]{DM03}%
  \BibitemOpen
  \bibfield  {author} {\bibinfo {author} {\bibfnamefont {Jonathan~P.}\
  \bibnamefont {Dowling}}\ and\ \bibinfo {author} {\bibfnamefont {Gerard~J.}\
  \bibnamefont {Milburn}},\ }\bibfield  {title} {\enquote {\bibinfo {title}
  {Quantum technology: the second quantum revolution},}\ }\href {\doibase
  10.1098/rsta.2003.1227} {\bibfield  {journal} {\bibinfo  {journal}
  {Philosophical Transactions of the Royal Society of London. Series A:
  Mathematical, Physical and Engineering Sciences}\ }\textbf {\bibinfo {volume}
  {361}},\ \bibinfo {pages} {1655--1674} (\bibinfo {year} {2003})},\ \bibinfo
  {note} {arXiv:quant-ph/0206091}\BibitemShut {NoStop}%
\bibitem [{\citenamefont {Wehner}\ \emph {et~al.}(2018)\citenamefont {Wehner},
  \citenamefont {Elkouss},\ and\ \citenamefont {Hanson}}]{Wehner2018}%
  \BibitemOpen
  \bibfield  {author} {\bibinfo {author} {\bibfnamefont {Stephanie}\
  \bibnamefont {Wehner}}, \bibinfo {author} {\bibfnamefont {David}\
  \bibnamefont {Elkouss}}, \ and\ \bibinfo {author} {\bibfnamefont {Ronald}\
  \bibnamefont {Hanson}},\ }\bibfield  {title} {\enquote {\bibinfo {title}
  {Quantum internet: A vision for the road ahead},}\ }\href {\doibase
  10.1126/science.aam9288} {\bibfield  {journal} {\bibinfo  {journal}
  {Science}\ }\textbf {\bibinfo {volume} {362}},\ \bibinfo {pages} {eaam9288}
  (\bibinfo {year} {2018})}\BibitemShut {NoStop}%
\bibitem [{\citenamefont {D\"{u}r}\ \emph {et~al.}(1999)\citenamefont
  {D\"{u}r}, \citenamefont {Briegel}, \citenamefont {Cirac},\ and\
  \citenamefont {Zoller}}]{Dr1999}%
  \BibitemOpen
  \bibfield  {author} {\bibinfo {author} {\bibfnamefont {W.}~\bibnamefont
  {D\"{u}r}}, \bibinfo {author} {\bibfnamefont {H.-J.}\ \bibnamefont
  {Briegel}}, \bibinfo {author} {\bibfnamefont {J.~I.}\ \bibnamefont {Cirac}},
  \ and\ \bibinfo {author} {\bibfnamefont {P.}~\bibnamefont {Zoller}},\
  }\bibfield  {title} {\enquote {\bibinfo {title} {Quantum repeaters based on
  entanglement purification},}\ }\href {\doibase 10.1103/physreva.59.169}
  {\bibfield  {journal} {\bibinfo  {journal} {Physical Review A}\ }\textbf
  {\bibinfo {volume} {59}},\ \bibinfo {pages} {169--181} (\bibinfo {year}
  {1999})}\BibitemShut {NoStop}%
\bibitem [{\citenamefont {Muralidharan}\ \emph {et~al.}(2016)\citenamefont
  {Muralidharan}, \citenamefont {Li}, \citenamefont {Kim}, \citenamefont
  {L\"{u}tkenhaus}, \citenamefont {Lukin},\ and\ \citenamefont
  {Jiang}}]{Muralidharan2016}%
  \BibitemOpen
  \bibfield  {author} {\bibinfo {author} {\bibfnamefont {Sreraman}\
  \bibnamefont {Muralidharan}}, \bibinfo {author} {\bibfnamefont {Linshu}\
  \bibnamefont {Li}}, \bibinfo {author} {\bibfnamefont {Jungsang}\ \bibnamefont
  {Kim}}, \bibinfo {author} {\bibfnamefont {Norbert}\ \bibnamefont
  {L\"{u}tkenhaus}}, \bibinfo {author} {\bibfnamefont {Mikhail~D.}\
  \bibnamefont {Lukin}}, \ and\ \bibinfo {author} {\bibfnamefont {Liang}\
  \bibnamefont {Jiang}},\ }\bibfield  {title} {\enquote {\bibinfo {title}
  {Optimal architectures for long distance quantum communication},}\ }\href
  {\doibase 10.1038/srep20463} {\bibfield  {journal} {\bibinfo  {journal}
  {Scientific Reports}\ }\textbf {\bibinfo {volume} {6}} (\bibinfo {year}
  {2016}),\ 10.1038/srep20463}\BibitemShut {NoStop}%
\bibitem [{\citenamefont {Zhang}\ \emph {et~al.}(2018)\citenamefont {Zhang},
  \citenamefont {Xu}, \citenamefont {Chen}, \citenamefont {Peng},\ and\
  \citenamefont {Pan}}]{ZXC+18}%
  \BibitemOpen
  \bibfield  {author} {\bibinfo {author} {\bibfnamefont {Qiang}\ \bibnamefont
  {Zhang}}, \bibinfo {author} {\bibfnamefont {Feihu}\ \bibnamefont {Xu}},
  \bibinfo {author} {\bibfnamefont {Yu-Ao}\ \bibnamefont {Chen}}, \bibinfo
  {author} {\bibfnamefont {Cheng-Zhi}\ \bibnamefont {Peng}}, \ and\ \bibinfo
  {author} {\bibfnamefont {Jian-Wei}\ \bibnamefont {Pan}},\ }\bibfield  {title}
  {\enquote {\bibinfo {title} {Large scale quantum key distribution: challenges
  and solutions [invited]},}\ }\href {\doibase 10.1364/oe.26.024260} {\bibfield
   {journal} {\bibinfo  {journal} {Optics Express}\ }\textbf {\bibinfo {volume}
  {26}},\ \bibinfo {pages} {24260} (\bibinfo {year} {2018})},\ \bibinfo {note}
  {arXiv:1809.02291}\BibitemShut {NoStop}%
\bibitem [{\citenamefont {Bennett}\ and\ \citenamefont
  {Brassard}(1984)}]{BB84}%
  \BibitemOpen
  \bibfield  {author} {\bibinfo {author} {\bibfnamefont {Charles~H.}\
  \bibnamefont {Bennett}}\ and\ \bibinfo {author} {\bibfnamefont {Gilles}\
  \bibnamefont {Brassard}},\ }\bibfield  {title} {\enquote {\bibinfo {title}
  {Quantum cryptography: Public key distribution and coin tossing},}\ }in\
  \href@noop {} {\emph {\bibinfo {booktitle} {International Conference on
  Computer System and Signal Processing, IEEE, 1984}}}\ (\bibinfo {year}
  {1984})\ pp.\ \bibinfo {pages} {175--179}\BibitemShut {NoStop}%
\bibitem [{\citenamefont {Pironio}\ \emph {et~al.}(2009)\citenamefont
  {Pironio}, \citenamefont {Ac{\'{\i}}n}, \citenamefont {Brunner},
  \citenamefont {Gisin}, \citenamefont {Massar},\ and\ \citenamefont
  {Scarani}}]{Pironio2009}%
  \BibitemOpen
  \bibfield  {author} {\bibinfo {author} {\bibfnamefont {Stefano}\ \bibnamefont
  {Pironio}}, \bibinfo {author} {\bibfnamefont {Antonio}\ \bibnamefont
  {Ac{\'{\i}}n}}, \bibinfo {author} {\bibfnamefont {Nicolas}\ \bibnamefont
  {Brunner}}, \bibinfo {author} {\bibfnamefont {Nicolas}\ \bibnamefont
  {Gisin}}, \bibinfo {author} {\bibfnamefont {Serge}\ \bibnamefont {Massar}}, \
  and\ \bibinfo {author} {\bibfnamefont {Valerio}\ \bibnamefont {Scarani}},\
  }\bibfield  {title} {\enquote {\bibinfo {title} {Device-independent quantum
  key distribution secure against collective attacks},}\ }\href {\doibase
  10.1088/1367-2630/11/4/045021} {\bibfield  {journal} {\bibinfo  {journal}
  {New Journal of Physics}\ }\textbf {\bibinfo {volume} {11}},\ \bibinfo
  {pages} {045021} (\bibinfo {year} {2009})}\BibitemShut {NoStop}%
\bibitem [{\citenamefont {Becker}\ \emph {et~al.}(2013)\citenamefont {Becker},
  \citenamefont {Regazzoni}, \citenamefont {Paar},\ and\ \citenamefont
  {Burleson}}]{Becker2013}%
  \BibitemOpen
  \bibfield  {author} {\bibinfo {author} {\bibfnamefont {Georg~T.}\
  \bibnamefont {Becker}}, \bibinfo {author} {\bibfnamefont {Francesco}\
  \bibnamefont {Regazzoni}}, \bibinfo {author} {\bibfnamefont {Christof}\
  \bibnamefont {Paar}}, \ and\ \bibinfo {author} {\bibfnamefont {Wayne~P.}\
  \bibnamefont {Burleson}},\ }\bibfield  {title} {\enquote {\bibinfo {title}
  {Stealthy dopant-level hardware trojans},}\ }in\ \href {\doibase
  10.1007/978-3-642-40349-1_12} {\emph {\bibinfo {booktitle} {Cryptographic
  Hardware and Embedded Systems - {CHES} 2013}}}\ (\bibinfo  {publisher}
  {Springer Berlin Heidelberg},\ \bibinfo {year} {2013})\ pp.\ \bibinfo {pages}
  {197--214}\BibitemShut {NoStop}%
\bibitem [{\citenamefont {Makarov}(2009)}]{Makarov2009}%
  \BibitemOpen
  \bibfield  {author} {\bibinfo {author} {\bibfnamefont {Vadim}\ \bibnamefont
  {Makarov}},\ }\bibfield  {title} {\enquote {\bibinfo {title} {Controlling
  passively quenched single photon detectors by bright light},}\ }\href
  {\doibase 10.1088/1367-2630/11/6/065003} {\bibfield  {journal} {\bibinfo
  {journal} {New Journal of Physics}\ }\textbf {\bibinfo {volume} {11}},\
  \bibinfo {pages} {065003} (\bibinfo {year} {2009})}\BibitemShut {NoStop}%
\bibitem [{\citenamefont {Ekert}(1991)}]{E91}%
  \BibitemOpen
  \bibfield  {author} {\bibinfo {author} {\bibfnamefont {Artur~K.}\
  \bibnamefont {Ekert}},\ }\bibfield  {title} {\enquote {\bibinfo {title}
  {Quantum cryptography based on bell's theorem},}\ }\href {\doibase
  10.1103/physrevlett.67.661} {\bibfield  {journal} {\bibinfo  {journal}
  {Physical Review Letters}\ }\textbf {\bibinfo {volume} {67}},\ \bibinfo
  {pages} {661--663} (\bibinfo {year} {1991})}\BibitemShut {NoStop}%
\bibitem [{\citenamefont {Brunner}\ \emph {et~al.}(2014)\citenamefont
  {Brunner}, \citenamefont {Cavalcanti}, \citenamefont {Pironio}, \citenamefont
  {Scarani},\ and\ \citenamefont {Wehner}}]{Bell-nonlocality}%
  \BibitemOpen
  \bibfield  {author} {\bibinfo {author} {\bibfnamefont {N.}~\bibnamefont
  {Brunner}}, \bibinfo {author} {\bibfnamefont {D.}~\bibnamefont {Cavalcanti}},
  \bibinfo {author} {\bibfnamefont {S.}~\bibnamefont {Pironio}}, \bibinfo
  {author} {\bibfnamefont {V.}~\bibnamefont {Scarani}}, \ and\ \bibinfo
  {author} {\bibfnamefont {S.}~\bibnamefont {Wehner}},\ }\bibfield  {title}
  {\enquote {\bibinfo {title} {Bell nonlocality},}\ }\href@noop {} {\bibfield
  {journal} {\bibinfo  {journal} {Rev. Mod. Phys.}\ }\textbf {\bibinfo {volume}
  {86}},\ \bibinfo {pages} {839} (\bibinfo {year} {2014})},\ \Eprint
  {http://arxiv.org/abs/quant-ph/1303.2849} {quant-ph/1303.2849} \BibitemShut
  {NoStop}%
\bibitem [{\citenamefont {Zhang}\ \emph {et~al.}(2021)\citenamefont {Zhang},
  \citenamefont {van Leent}, \citenamefont {Redeker}, \citenamefont {Garthoff},
  \citenamefont {Schwonnek}, \citenamefont {Fertig}, \citenamefont {Eppelt},
  \citenamefont {Scarani}, \citenamefont {Lim},\ and\ \citenamefont
  {Weinfurter}}]{Harald_exp}%
  \BibitemOpen
  \bibfield  {author} {\bibinfo {author} {\bibfnamefont {Wei}\ \bibnamefont
  {Zhang}}, \bibinfo {author} {\bibfnamefont {Tim}\ \bibnamefont {van Leent}},
  \bibinfo {author} {\bibfnamefont {Kai}\ \bibnamefont {Redeker}}, \bibinfo
  {author} {\bibfnamefont {Robert}\ \bibnamefont {Garthoff}}, \bibinfo {author}
  {\bibfnamefont {Rene}\ \bibnamefont {Schwonnek}}, \bibinfo {author}
  {\bibfnamefont {Florian}\ \bibnamefont {Fertig}}, \bibinfo {author}
  {\bibfnamefont {Sebastian}\ \bibnamefont {Eppelt}}, \bibinfo {author}
  {\bibfnamefont {Valerio}\ \bibnamefont {Scarani}}, \bibinfo {author}
  {\bibfnamefont {Charles C.~W.}\ \bibnamefont {Lim}}, \ and\ \bibinfo {author}
  {\bibfnamefont {Harald}\ \bibnamefont {Weinfurter}},\ }\href@noop {}
  {\enquote {\bibinfo {title} {Experimental device-independent quantum key
  distribution between distant users},}\ } (\bibinfo {year} {2021}),\ \Eprint
  {http://arxiv.org/abs/2110.00575} {arXiv:2110.00575 [quant-ph]} \BibitemShut
  {NoStop}%
\bibitem [{\citenamefont {Nadlinger}\ \emph {et~al.}(2021)\citenamefont
  {Nadlinger}, \citenamefont {Drmota}, \citenamefont {Nichol}, \citenamefont
  {Araneda}, \citenamefont {Main}, \citenamefont {Srinivas}, \citenamefont
  {Lucas}, \citenamefont {Ballance}, \citenamefont {Ivanov}, \citenamefont
  {Tan}, \citenamefont {Sekatski}, \citenamefont {Urbanke}, \citenamefont
  {Renner}, \citenamefont {Sangouard},\ and\ \citenamefont
  {Bancal}}]{Renner_exp}%
  \BibitemOpen
  \bibfield  {author} {\bibinfo {author} {\bibfnamefont {D.~P.}\ \bibnamefont
  {Nadlinger}}, \bibinfo {author} {\bibfnamefont {P.}~\bibnamefont {Drmota}},
  \bibinfo {author} {\bibfnamefont {B.~C.}\ \bibnamefont {Nichol}}, \bibinfo
  {author} {\bibfnamefont {G.}~\bibnamefont {Araneda}}, \bibinfo {author}
  {\bibfnamefont {D.}~\bibnamefont {Main}}, \bibinfo {author} {\bibfnamefont
  {R.}~\bibnamefont {Srinivas}}, \bibinfo {author} {\bibfnamefont {D.~M.}\
  \bibnamefont {Lucas}}, \bibinfo {author} {\bibfnamefont {C.~J.}\ \bibnamefont
  {Ballance}}, \bibinfo {author} {\bibfnamefont {K.}~\bibnamefont {Ivanov}},
  \bibinfo {author} {\bibfnamefont {E.~Y-Z.}\ \bibnamefont {Tan}}, \bibinfo
  {author} {\bibfnamefont {P.}~\bibnamefont {Sekatski}}, \bibinfo {author}
  {\bibfnamefont {R.~L.}\ \bibnamefont {Urbanke}}, \bibinfo {author}
  {\bibfnamefont {R.}~\bibnamefont {Renner}}, \bibinfo {author} {\bibfnamefont
  {N.}~\bibnamefont {Sangouard}}, \ and\ \bibinfo {author} {\bibfnamefont
  {J-D.}\ \bibnamefont {Bancal}},\ }\href@noop {} {\enquote {\bibinfo {title}
  {Device-independent quantum key distribution},}\ } (\bibinfo {year} {2021}),\
  \Eprint {http://arxiv.org/abs/2109.14600} {arXiv:2109.14600 [quant-ph]}
  \BibitemShut {NoStop}%
\bibitem [{\citenamefont {Liu}\ \emph {et~al.}(2021)\citenamefont {Liu},
  \citenamefont {Zhang}, \citenamefont {Zhen}, \citenamefont {Li},
  \citenamefont {Liu}, \citenamefont {Fan}, \citenamefont {Xu}, \citenamefont
  {Zhang},\ and\ \citenamefont {Pan}}]{JWP_exp}%
  \BibitemOpen
  \bibfield  {author} {\bibinfo {author} {\bibfnamefont {Wen-Zhao}\
  \bibnamefont {Liu}}, \bibinfo {author} {\bibfnamefont {Yu-Zhe}\ \bibnamefont
  {Zhang}}, \bibinfo {author} {\bibfnamefont {Yi-Zheng}\ \bibnamefont {Zhen}},
  \bibinfo {author} {\bibfnamefont {Ming-Han}\ \bibnamefont {Li}}, \bibinfo
  {author} {\bibfnamefont {Yang}\ \bibnamefont {Liu}}, \bibinfo {author}
  {\bibfnamefont {Jingyun}\ \bibnamefont {Fan}}, \bibinfo {author}
  {\bibfnamefont {Feihu}\ \bibnamefont {Xu}}, \bibinfo {author} {\bibfnamefont
  {Qiang}\ \bibnamefont {Zhang}}, \ and\ \bibinfo {author} {\bibfnamefont
  {Jian-Wei}\ \bibnamefont {Pan}},\ }\href@noop {} {\enquote {\bibinfo {title}
  {High-speed device-independent quantum key distribution against collective
  attacks},}\ } (\bibinfo {year} {2021}),\ \Eprint
  {http://arxiv.org/abs/2110.01480} {arXiv:2110.01480 [quant-ph]} \BibitemShut
  {NoStop}%
\bibitem [{\citenamefont {Kaur}\ \emph {et~al.}(2020)\citenamefont {Kaur},
  \citenamefont {Wilde},\ and\ \citenamefont {Winter}}]{KW17}%
  \BibitemOpen
  \bibfield  {author} {\bibinfo {author} {\bibfnamefont {Eneet}\ \bibnamefont
  {Kaur}}, \bibinfo {author} {\bibfnamefont {Mark~M}\ \bibnamefont {Wilde}}, \
  and\ \bibinfo {author} {\bibfnamefont {Andreas}\ \bibnamefont {Winter}},\
  }\bibfield  {title} {\enquote {\bibinfo {title} {Fundamental limits on key
  rates in device-independent quantum key distribution},}\ }\href
  {https://doi.org/10.1088/1367-2630/ab6eaa} {\bibfield  {journal} {\bibinfo
  {journal} {New Journal of Physics}\ }\textbf {\bibinfo {volume} {22}},\
  \bibinfo {pages} {023039} (\bibinfo {year} {2020})}\BibitemShut {NoStop}%
\bibitem [{\citenamefont {Christandl}\ \emph {et~al.}(2021)\citenamefont
  {Christandl}, \citenamefont {Ferrara},\ and\ \citenamefont
  {Horodecki}}]{CFH21}%
  \BibitemOpen
  \bibfield  {author} {\bibinfo {author} {\bibfnamefont {Matthias}\
  \bibnamefont {Christandl}}, \bibinfo {author} {\bibfnamefont {Roberto}\
  \bibnamefont {Ferrara}}, \ and\ \bibinfo {author} {\bibfnamefont {Karol}\
  \bibnamefont {Horodecki}},\ }\bibfield  {title} {\enquote {\bibinfo {title}
  {Upper bounds on device-independent quantum key distribution},}\ }\href
  {\doibase 10.1103/physrevlett.126.160501} {\bibfield  {journal} {\bibinfo
  {journal} {Physical Review Letters}\ }\textbf {\bibinfo {volume} {126}}
  (\bibinfo {year} {2021}),\ 10.1103/physrevlett.126.160501}\BibitemShut
  {NoStop}%
\bibitem [{\citenamefont {Farkas}\ \emph {et~al.}(2021)\citenamefont {Farkas},
  \citenamefont {Balanz{\'{o}}-Juand{\'{o}}}, \citenamefont {{\L}ukanowski},
  \citenamefont {Ko{\l}ody{\'{n}}ski},\ and\ \citenamefont
  {Ac{\'{\i}}n}}]{FBL+21}%
  \BibitemOpen
  \bibfield  {author} {\bibinfo {author} {\bibfnamefont {M{\'{a}}t{\'{e}}}\
  \bibnamefont {Farkas}}, \bibinfo {author} {\bibfnamefont {Maria}\
  \bibnamefont {Balanz{\'{o}}-Juand{\'{o}}}}, \bibinfo {author} {\bibfnamefont
  {Karol}\ \bibnamefont {{\L}ukanowski}}, \bibinfo {author} {\bibfnamefont
  {Jan}\ \bibnamefont {Ko{\l}ody{\'{n}}ski}}, \ and\ \bibinfo {author}
  {\bibfnamefont {Antonio}\ \bibnamefont {Ac{\'{\i}}n}},\ }\bibfield  {title}
  {\enquote {\bibinfo {title} {Bell nonlocality is not sufficient for the
  security of standard device-independent quantum key distribution
  protocols},}\ }\href {\doibase 10.1103/physrevlett.127.050503} {\bibfield
  {journal} {\bibinfo  {journal} {Physical Review Letters}\ }\textbf {\bibinfo
  {volume} {127}} (\bibinfo {year} {2021}),\
  10.1103/physrevlett.127.050503}\BibitemShut {NoStop}%
\bibitem [{\citenamefont {{Kaur}}\ \emph {et~al.}(2021)\citenamefont {{Kaur}},
  \citenamefont {{Horodecki}},\ and\ \citenamefont {{Das}}}]{KHD21}%
  \BibitemOpen
  \bibfield  {author} {\bibinfo {author} {\bibfnamefont {Eneet}\ \bibnamefont
  {{Kaur}}}, \bibinfo {author} {\bibfnamefont {Karol}\ \bibnamefont
  {{Horodecki}}}, \ and\ \bibinfo {author} {\bibfnamefont {Siddhartha}\
  \bibnamefont {{Das}}},\ }\bibfield  {title} {\enquote {\bibinfo {title}
  {Upper bounds on device-independent quantum key distribution rates in static
  and dynamic scenarios},}\ }\href {https://arxiv.org/abs/2107.06411} {\
  (\bibinfo {year} {2021})},\ \bibinfo {note} {arXiv:2107.06411}\BibitemShut
  {NoStop}%
\bibitem [{\citenamefont {Murta}\ \emph {et~al.}(2020)\citenamefont {Murta},
  \citenamefont {Grasselli}, \citenamefont {Kampermann},\ and\ \citenamefont
  {Bru{\ss}}}]{Murta2020}%
  \BibitemOpen
  \bibfield  {author} {\bibinfo {author} {\bibfnamefont {Gl{\'{a}}ucia}\
  \bibnamefont {Murta}}, \bibinfo {author} {\bibfnamefont {Federico}\
  \bibnamefont {Grasselli}}, \bibinfo {author} {\bibfnamefont {Hermann}\
  \bibnamefont {Kampermann}}, \ and\ \bibinfo {author} {\bibfnamefont {Dagmar}\
  \bibnamefont {Bru{\ss}}},\ }\bibfield  {title} {\enquote {\bibinfo {title}
  {Quantum conference key agreement: A review},}\ }\href {\doibase
  10.1002/qute.202000025} {\bibfield  {journal} {\bibinfo  {journal} {Advanced
  Quantum Technologies}\ }\textbf {\bibinfo {volume} {3}},\ \bibinfo {pages}
  {2000025} (\bibinfo {year} {2020})}\BibitemShut {NoStop}%
\bibitem [{\citenamefont {Ribeiro}\ \emph {et~al.}(2019)\citenamefont
  {Ribeiro}, \citenamefont {Murta},\ and\ \citenamefont {Wehner}}]{RMW18}%
  \BibitemOpen
  \bibfield  {author} {\bibinfo {author} {\bibfnamefont {J\'er\'emy}\
  \bibnamefont {Ribeiro}}, \bibinfo {author} {\bibfnamefont {Gl\'aucia}\
  \bibnamefont {Murta}}, \ and\ \bibinfo {author} {\bibfnamefont {Stephanie}\
  \bibnamefont {Wehner}},\ }\bibfield  {title} {\enquote {\bibinfo {title}
  {Fully device-independent conference key agreement},}\ }\href@noop {} {\
  (\bibinfo {year} {2019})},\ \bibinfo {note} {arXiv:1708.00798v2}\BibitemShut
  {NoStop}%
\bibitem [{\citenamefont {Das}\ \emph {et~al.}(2021)\citenamefont {Das},
  \citenamefont {B\"auml}, \citenamefont {Winczewski},\ and\ \citenamefont
  {Horodecki}}]{DBWH19}%
  \BibitemOpen
  \bibfield  {author} {\bibinfo {author} {\bibfnamefont {Siddhartha}\
  \bibnamefont {Das}}, \bibinfo {author} {\bibfnamefont {Stefan}\ \bibnamefont
  {B\"auml}}, \bibinfo {author} {\bibfnamefont {Marek}\ \bibnamefont
  {Winczewski}}, \ and\ \bibinfo {author} {\bibfnamefont {Karol}\ \bibnamefont
  {Horodecki}},\ }\bibfield  {title} {\enquote {\bibinfo {title} {Universal
  limitations on quantum key distribution over a network},}\ }\href {\doibase
  10.1103/PhysRevX.11.041016} {\bibfield  {journal} {\bibinfo  {journal}
  {Physical Review X}\ }\textbf {\bibinfo {volume} {11}},\ \bibinfo {pages}
  {041016} (\bibinfo {year} {2021})},\ \bibinfo {note}
  {arXiv:1912.03646}\BibitemShut {NoStop}%
\bibitem [{\citenamefont {Arnon-Friedman}\ and\ \citenamefont
  {Leditzky}(2021)}]{AL20}%
  \BibitemOpen
  \bibfield  {author} {\bibinfo {author} {\bibfnamefont {Rotem}\ \bibnamefont
  {Arnon-Friedman}}\ and\ \bibinfo {author} {\bibfnamefont {Felix}\
  \bibnamefont {Leditzky}},\ }\bibfield  {title} {\enquote {\bibinfo {title}
  {Upper bounds on device-independent quantum key distribution rates and a
  revised peres conjecture},}\ }\href {\doibase 10.1109/TIT.2021.3086505}
  {\bibfield  {journal} {\bibinfo  {journal} {IEEE Transactions on Information
  Theory}\ }\textbf {\bibinfo {volume} {67}},\ \bibinfo {pages} {6606--6618}
  (\bibinfo {year} {2021})},\ \bibinfo {note} {arXiv:2005.12325}\BibitemShut
  {NoStop}%
\bibitem [{\citenamefont {Yang}\ \emph {et~al.}(2009)\citenamefont {Yang},
  \citenamefont {Horodecki}, \citenamefont {Horodecki}, \citenamefont
  {Horodecki}, \citenamefont {Oppenheim},\ and\ \citenamefont
  {Song}}]{Yang2009}%
  \BibitemOpen
  \bibfield  {author} {\bibinfo {author} {\bibfnamefont {Dong}\ \bibnamefont
  {Yang}}, \bibinfo {author} {\bibfnamefont {Karol}\ \bibnamefont {Horodecki}},
  \bibinfo {author} {\bibfnamefont {Michal}\ \bibnamefont {Horodecki}},
  \bibinfo {author} {\bibfnamefont {Pawel}\ \bibnamefont {Horodecki}}, \bibinfo
  {author} {\bibfnamefont {Jonathan}\ \bibnamefont {Oppenheim}}, \ and\
  \bibinfo {author} {\bibfnamefont {Wei}\ \bibnamefont {Song}},\ }\bibfield
  {title} {\enquote {\bibinfo {title} {Squashed entanglement for multipartite
  states and entanglement measures based on the mixed convex roof},}\ }\href
  {\doibase 10.1109/tit.2009.2021373} {\bibfield  {journal} {\bibinfo
  {journal} {{IEEE} Transactions on Information Theory}\ }\textbf {\bibinfo
  {volume} {55}},\ \bibinfo {pages} {3375--3387} (\bibinfo {year}
  {2009})}\BibitemShut {NoStop}%
\bibitem [{\citenamefont {Christandl}\ \emph {et~al.}(2007)\citenamefont
  {Christandl}, \citenamefont {Ekert}, \citenamefont {Horodecki}, \citenamefont
  {Horodecki}, \citenamefont {Oppenheim},\ and\ \citenamefont
  {Renner}}]{CEHHOR}%
  \BibitemOpen
  \bibfield  {author} {\bibinfo {author} {\bibfnamefont {Matthias}\
  \bibnamefont {Christandl}}, \bibinfo {author} {\bibfnamefont {Artur}\
  \bibnamefont {Ekert}}, \bibinfo {author} {\bibfnamefont {Micha{\l}}\
  \bibnamefont {Horodecki}}, \bibinfo {author} {\bibfnamefont {Pawe{\l}}\
  \bibnamefont {Horodecki}}, \bibinfo {author} {\bibfnamefont {Jonathan}\
  \bibnamefont {Oppenheim}}, \ and\ \bibinfo {author} {\bibfnamefont {Renato}\
  \bibnamefont {Renner}},\ }\bibfield  {title} {\enquote {\bibinfo {title}
  {Unifying classical and quantum key distillation},}\ }in\ \href {\doibase
  10.1007/978-3-540-70936-7_25} {\emph {\bibinfo {booktitle} {Theory of
  Cryptography}}}\ (\bibinfo  {publisher} {Springer Berlin Heidelberg},\
  \bibinfo {year} {2007})\ pp.\ \bibinfo {pages} {456--478}\BibitemShut
  {NoStop}%
\bibitem [{\citenamefont {Cerf}\ \emph {et~al.}(2002)\citenamefont {Cerf},
  \citenamefont {Massar},\ and\ \citenamefont {Schneider}}]{Cerf2002}%
  \BibitemOpen
  \bibfield  {author} {\bibinfo {author} {\bibfnamefont {N.~J.}\ \bibnamefont
  {Cerf}}, \bibinfo {author} {\bibfnamefont {S.}~\bibnamefont {Massar}}, \ and\
  \bibinfo {author} {\bibfnamefont {S.}~\bibnamefont {Schneider}},\ }\bibfield
  {title} {\enquote {\bibinfo {title} {Multipartite classical and quantum
  secrecy monotones},}\ }\href {\doibase 10.1103/physreva.66.042309} {\bibfield
   {journal} {\bibinfo  {journal} {Physical Review A}\ }\textbf {\bibinfo
  {volume} {66}},\ \bibinfo {pages} {042309} (\bibinfo {year}
  {2002})}\BibitemShut {NoStop}%
\bibitem [{\citenamefont {Horodecki}\ \emph {et~al.}(2009)\citenamefont
  {Horodecki}, \citenamefont {Horodecki}, \citenamefont {Horodecki},\ and\
  \citenamefont {Horodecki}}]{Horodecki2009}%
  \BibitemOpen
  \bibfield  {author} {\bibinfo {author} {\bibfnamefont {Ryszard}\ \bibnamefont
  {Horodecki}}, \bibinfo {author} {\bibfnamefont {Pawe{\l}}\ \bibnamefont
  {Horodecki}}, \bibinfo {author} {\bibfnamefont {Micha{\l}}\ \bibnamefont
  {Horodecki}}, \ and\ \bibinfo {author} {\bibfnamefont {Karol}\ \bibnamefont
  {Horodecki}},\ }\bibfield  {title} {\enquote {\bibinfo {title} {Quantum
  entanglement},}\ }\href {\doibase 10.1103/revmodphys.81.865} {\bibfield
  {journal} {\bibinfo  {journal} {Reviews of Modern Physics}\ }\textbf
  {\bibinfo {volume} {81}},\ \bibinfo {pages} {865--942} (\bibinfo {year}
  {2009})}\BibitemShut {NoStop}%
\bibitem [{\citenamefont {Arnon-Friedman}\ \emph {et~al.}(2018)\citenamefont
  {Arnon-Friedman}, \citenamefont {Dupuis}, \citenamefont {Fawzi},
  \citenamefont {Renner},\ and\ \citenamefont {Vidick}}]{AFFRV18}%
  \BibitemOpen
  \bibfield  {author} {\bibinfo {author} {\bibfnamefont {Rotem}\ \bibnamefont
  {Arnon-Friedman}}, \bibinfo {author} {\bibfnamefont {Fr{\'{e}}d{\'{e}}ric}\
  \bibnamefont {Dupuis}}, \bibinfo {author} {\bibfnamefont {Omar}\ \bibnamefont
  {Fawzi}}, \bibinfo {author} {\bibfnamefont {Renato}\ \bibnamefont {Renner}},
  \ and\ \bibinfo {author} {\bibfnamefont {Thomas}\ \bibnamefont {Vidick}},\
  }\bibfield  {title} {\enquote {\bibinfo {title} {Practical device-independent
  quantum cryptography via entropy accumulation},}\ }\href {\doibase
  10.1038/s41467-017-02307-4} {\ \textbf {\bibinfo {volume} {9}} (\bibinfo
  {year} {2018}),\ 10.1038/s41467-017-02307-4}\BibitemShut {NoStop}%
\bibitem [{\citenamefont {Davis}\ \emph {et~al.}(2018)\citenamefont {Davis},
  \citenamefont {Shirokov},\ and\ \citenamefont {Wilde}}]{DSW18}%
  \BibitemOpen
  \bibfield  {author} {\bibinfo {author} {\bibfnamefont {Noah}\ \bibnamefont
  {Davis}}, \bibinfo {author} {\bibfnamefont {Maksim~E.}\ \bibnamefont
  {Shirokov}}, \ and\ \bibinfo {author} {\bibfnamefont {Mark~M.}\ \bibnamefont
  {Wilde}},\ }\bibfield  {title} {\enquote {\bibinfo {title}
  {Energy-constrained two-way assisted private and quantum capacities of
  quantum channels},}\ }\href {\doibase 10.1103/PhysRevA.97.062310} {\bibfield
  {journal} {\bibinfo  {journal} {Physical Review A}\ }\textbf {\bibinfo
  {volume} {97}},\ \bibinfo {pages} {062310} (\bibinfo {year}
  {2018})}\BibitemShut {NoStop}%
\bibitem [{\citenamefont {Horodecki}\ \emph {et~al.}(2022)\citenamefont
  {Horodecki}, \citenamefont {Winczewski},\ and\ \citenamefont {Das}}]{HWD22}%
  \BibitemOpen
  \bibfield  {author} {\bibinfo {author} {\bibfnamefont {Karol}\ \bibnamefont
  {Horodecki}}, \bibinfo {author} {\bibfnamefont {Marek}\ \bibnamefont
  {Winczewski}}, \ and\ \bibinfo {author} {\bibfnamefont {Siddhartha}\
  \bibnamefont {Das}},\ }\bibfield  {title} {\enquote {\bibinfo {title}
  {Fundamental limitations on the device-independent quantum conference key
  agreement},}\ }\href {\doibase 10.1103/PhysRevA.105.022604} {\bibfield
  {journal} {\bibinfo  {journal} {Physical Review A}\ }\textbf {\bibinfo
  {volume} {105}},\ \bibinfo {pages} {022604} (\bibinfo {year}
  {2022})}\BibitemShut {NoStop}%
\bibitem [{\citenamefont {Mermin}(1990)}]{Mer90}%
  \BibitemOpen
  \bibfield  {author} {\bibinfo {author} {\bibfnamefont {N.~David}\
  \bibnamefont {Mermin}},\ }\bibfield  {title} {\enquote {\bibinfo {title}
  {Extreme quantum entanglement in a superposition of macroscopically distinct
  states},}\ }\href {\doibase 10.1103/PhysRevLett.65.1838} {\bibfield
  {journal} {\bibinfo  {journal} {Physical Review Letters}\ }\textbf {\bibinfo
  {volume} {65}},\ \bibinfo {pages} {1838--1840} (\bibinfo {year}
  {1990})}\BibitemShut {NoStop}%
\bibitem [{\citenamefont {Ardehali}(1992)}]{Ard92}%
  \BibitemOpen
  \bibfield  {author} {\bibinfo {author} {\bibfnamefont {M.}~\bibnamefont
  {Ardehali}},\ }\bibfield  {title} {\enquote {\bibinfo {title} {Bell
  inequalities with a magnitude of violation that grows exponentially with the
  number of particles},}\ }\href {\doibase 10.1103/PhysRevA.46.5375} {\bibfield
   {journal} {\bibinfo  {journal} {Physical Review A}\ }\textbf {\bibinfo
  {volume} {46}},\ \bibinfo {pages} {5375--5378} (\bibinfo {year}
  {1992})}\BibitemShut {NoStop}%
\bibitem [{\citenamefont {Belinski{\u\i}}\ and\ \citenamefont
  {Klyshko}(1993)}]{BK93}%
  \BibitemOpen
  \bibfield  {author} {\bibinfo {author} {\bibfnamefont {AV}~\bibnamefont
  {Belinski{\u\i}}}\ and\ \bibinfo {author} {\bibfnamefont {David~Nikolaevich}\
  \bibnamefont {Klyshko}},\ }\bibfield  {title} {\enquote {\bibinfo {title}
  {Interference of light and {B}ell's theorem},}\ }\href@noop {} {\bibfield
  {journal} {\bibinfo  {journal} {Physics-Uspekhi}\ }\textbf {\bibinfo {volume}
  {36}},\ \bibinfo {pages} {653} (\bibinfo {year} {1993})}\BibitemShut
  {NoStop}%
\bibitem [{\citenamefont {Seevinck}\ and\ \citenamefont
  {Svetlichny}(2002)}]{SS02}%
  \BibitemOpen
  \bibfield  {author} {\bibinfo {author} {\bibfnamefont {Michael}\ \bibnamefont
  {Seevinck}}\ and\ \bibinfo {author} {\bibfnamefont {George}\ \bibnamefont
  {Svetlichny}},\ }\bibfield  {title} {\enquote {\bibinfo {title} {Bell-type
  inequalities for partial separability in $n$-particle systems and quantum
  mechanical violations},}\ }\href {\doibase 10.1103/PhysRevLett.89.060401}
  {\bibfield  {journal} {\bibinfo  {journal} {Physical Review Letters}\
  }\textbf {\bibinfo {volume} {89}},\ \bibinfo {pages} {060401} (\bibinfo
  {year} {2002})}\BibitemShut {NoStop}%
\bibitem [{\citenamefont {\.{Z}ukowski}\ and\ \citenamefont
  {Brukner}(2002)}]{ZB02}%
  \BibitemOpen
  \bibfield  {author} {\bibinfo {author} {\bibfnamefont {Marek}\ \bibnamefont
  {\.{Z}ukowski}}\ and\ \bibinfo {author} {\bibfnamefont {{\v{C}}aslav}\
  \bibnamefont {Brukner}},\ }\bibfield  {title} {\enquote {\bibinfo {title}
  {Bell's theorem for general {N}-qubit states},}\ }\href {\doibase
  10.1103/PhysRevLett.88.210401} {\bibfield  {journal} {\bibinfo  {journal}
  {Physical Review Letters}\ }\textbf {\bibinfo {volume} {88}},\ \bibinfo
  {pages} {210401} (\bibinfo {year} {2002})}\BibitemShut {NoStop}%
\bibitem [{\citenamefont {Werner}\ and\ \citenamefont {Wolf}(2001)}]{WW02}%
  \BibitemOpen
  \bibfield  {author} {\bibinfo {author} {\bibfnamefont {R.~F.}\ \bibnamefont
  {Werner}}\ and\ \bibinfo {author} {\bibfnamefont {M.~M.}\ \bibnamefont
  {Wolf}},\ }\bibfield  {title} {\enquote {\bibinfo {title} {All-multipartite
  {B}ell-correlation inequalities for two dichotomic observables per site},}\
  }\href {\doibase 10.1103/PhysRevA.64.032112} {\bibfield  {journal} {\bibinfo
  {journal} {Physical Review A}\ }\textbf {\bibinfo {volume} {64}},\ \bibinfo
  {pages} {032112} (\bibinfo {year} {2001})}\BibitemShut {NoStop}%
\bibitem [{\citenamefont {Yu}\ \emph {et~al.}(2012)\citenamefont {Yu},
  \citenamefont {Chen}, \citenamefont {Zhang}, \citenamefont {Lai},\ and\
  \citenamefont {Oh}}]{YCLO12}%
  \BibitemOpen
  \bibfield  {author} {\bibinfo {author} {\bibfnamefont {Sixia}\ \bibnamefont
  {Yu}}, \bibinfo {author} {\bibfnamefont {Qing}\ \bibnamefont {Chen}},
  \bibinfo {author} {\bibfnamefont {Chengjie}\ \bibnamefont {Zhang}}, \bibinfo
  {author} {\bibfnamefont {C.~H.}\ \bibnamefont {Lai}}, \ and\ \bibinfo
  {author} {\bibfnamefont {C.~H.}\ \bibnamefont {Oh}},\ }\bibfield  {title}
  {\enquote {\bibinfo {title} {All entangled pure states violate a single
  {B}ell's inequality},}\ }\href {\doibase 10.1103/PhysRevLett.109.120402}
  {\bibfield  {journal} {\bibinfo  {journal} {Physical Review Letters}\
  }\textbf {\bibinfo {volume} {109}},\ \bibinfo {pages} {120402} (\bibinfo
  {year} {2012})}\BibitemShut {NoStop}%
\bibitem [{\citenamefont {Home}\ \emph {et~al.}(2015)\citenamefont {Home},
  \citenamefont {Saha},\ and\ \citenamefont {Das}}]{HSD15}%
  \BibitemOpen
  \bibfield  {author} {\bibinfo {author} {\bibfnamefont {Dipankar}\
  \bibnamefont {Home}}, \bibinfo {author} {\bibfnamefont {Debashis}\
  \bibnamefont {Saha}}, \ and\ \bibinfo {author} {\bibfnamefont {Siddhartha}\
  \bibnamefont {Das}},\ }\bibfield  {title} {\enquote {\bibinfo {title}
  {Multipartite {B}ell-type inequality by generalizing {W}igner's argument},}\
  }\href {\doibase 10.1103/PhysRevA.91.012102} {\bibfield  {journal} {\bibinfo
  {journal} {Physical Review A}\ }\textbf {\bibinfo {volume} {91}},\ \bibinfo
  {pages} {012102} (\bibinfo {year} {2015})},\ \bibinfo {note}
  {arXiv:1410.7936}\BibitemShut {NoStop}%
\bibitem [{\citenamefont {Luo}(2021)}]{Luo21}%
  \BibitemOpen
  \bibfield  {author} {\bibinfo {author} {\bibfnamefont {Ming-Xing}\
  \bibnamefont {Luo}},\ }\bibfield  {title} {\enquote {\bibinfo {title} {Fully
  device-independent model on quantum networks},}\ }\href@noop {} {\  (\bibinfo
  {year} {2021})},\ \bibinfo {note} {arXiv:2106.15840}\BibitemShut {NoStop}%
\bibitem [{\citenamefont {\.{Z}ukowski}\ \emph {et~al.}(2002)\citenamefont
  {\.{Z}ukowski}, \citenamefont {Brukner}, \citenamefont {Laskowski},\ and\
  \citenamefont {Wie\ifmmode~\acute{s}\else \'{s}\fi{}niak}}]{ZBLW02}%
  \BibitemOpen
  \bibfield  {author} {\bibinfo {author} {\bibfnamefont {Marek}\ \bibnamefont
  {\.{Z}ukowski}}, \bibinfo {author} {\bibfnamefont {{\v{C}}aslav}\
  \bibnamefont {Brukner}}, \bibinfo {author} {\bibfnamefont {Wies\l{}aw}\
  \bibnamefont {Laskowski}}, \ and\ \bibinfo {author} {\bibfnamefont {Marcin}\
  \bibnamefont {Wie\ifmmode~\acute{s}\else \'{s}\fi{}niak}},\ }\bibfield
  {title} {\enquote {\bibinfo {title} {Do all pure entangled states violate
  {B}ell's inequalities for correlation functions?}}\ }\href {\doibase
  10.1103/PhysRevLett.88.210402} {\bibfield  {journal} {\bibinfo  {journal}
  {Physical Review Letters}\ }\textbf {\bibinfo {volume} {88}},\ \bibinfo
  {pages} {210402} (\bibinfo {year} {2002})}\BibitemShut {NoStop}%
\bibitem [{\citenamefont {Winczewski}\ \emph {et~al.}(2022)\citenamefont
  {Winczewski}, \citenamefont {Das},\ and\ \citenamefont {Horodecki}}]{WDH22}%
  \BibitemOpen
  \bibfield  {author} {\bibinfo {author} {\bibfnamefont {Marek}\ \bibnamefont
  {Winczewski}}, \bibinfo {author} {\bibfnamefont {Tamoghna}\ \bibnamefont
  {Das}}, \ and\ \bibinfo {author} {\bibfnamefont {Karol}\ \bibnamefont
  {Horodecki}},\ }\bibfield  {title} {\enquote {\bibinfo {title} {Limitations
  on a device-independent key secure against a nonsignaling adversary via
  squashed nonlocality},}\ }\href {\doibase 10.1103/physreva.106.052612}
  {\bibfield  {journal} {\bibinfo  {journal} {Physical Review A}\ }\textbf
  {\bibinfo {volume} {106}} (\bibinfo {year} {2022}),\
  10.1103/physreva.106.052612}\BibitemShut {NoStop}%
\bibitem [{\citenamefont {Watanabe}(1960)}]{Wat60}%
  \BibitemOpen
  \bibfield  {author} {\bibinfo {author} {\bibfnamefont {Satosi}\ \bibnamefont
  {Watanabe}},\ }\bibfield  {title} {\enquote {\bibinfo {title} {Information
  theoretical analysis of multivariate correlation},}\ }\href {\doibase
  10.1147/rd.41.0066} {\bibfield  {journal} {\bibinfo  {journal} {IBM Journal
  of Research and Development}\ }\textbf {\bibinfo {volume} {4}},\ \bibinfo
  {pages} {66--82} (\bibinfo {year} {1960})}\BibitemShut {NoStop}%
\bibitem [{\citenamefont {Shirokov}(2017)}]{Shi17}%
  \BibitemOpen
  \bibfield  {author} {\bibinfo {author} {\bibfnamefont {Maksim~E}\
  \bibnamefont {Shirokov}},\ }\bibfield  {title} {\enquote {\bibinfo {title}
  {Tight uniform continuity bounds for the quantum conditional mutual
  information, for the {Holevo} quantity, and for capacities of quantum
  channels},}\ }\href {https://doi.org/10.1063/1.4987135} {\bibfield  {journal}
  {\bibinfo  {journal} {Journal of Mathematical Physics}\ }\textbf {\bibinfo
  {volume} {58}},\ \bibinfo {pages} {102202} (\bibinfo {year}
  {2017})}\BibitemShut {NoStop}%
\bibitem [{\citenamefont {Gisin}\ and\ \citenamefont {Wolf}(2000)}]{GW00}%
  \BibitemOpen
  \bibfield  {author} {\bibinfo {author} {\bibfnamefont {Nicolas}\ \bibnamefont
  {Gisin}}\ and\ \bibinfo {author} {\bibfnamefont {Stefan}\ \bibnamefont
  {Wolf}},\ }\bibfield  {title} {\enquote {\bibinfo {title} {Linking classical
  and quantum key agreement: Is there {\textquotedblleft}bound
  information{\textquotedblright}?}}\ }in\ \href {\doibase
  10.1007/3-540-44598-6_30} {\emph {\bibinfo {booktitle} {Advances in
  Cryptology {\textemdash} {CRYPTO} 2000}}}\ (\bibinfo  {publisher} {Springer
  Berlin Heidelberg},\ \bibinfo {year} {2000})\ pp.\ \bibinfo {pages}
  {482--500}\BibitemShut {NoStop}%
\bibitem [{\citenamefont {Ac\'{\i}n}\ \emph {et~al.}(2006)\citenamefont
  {Ac\'{\i}n}, \citenamefont {Gisin},\ and\ \citenamefont {Masanes}}]{AGM+06}%
  \BibitemOpen
  \bibfield  {author} {\bibinfo {author} {\bibfnamefont {Antonio}\ \bibnamefont
  {Ac\'{\i}n}}, \bibinfo {author} {\bibfnamefont {Nicolas}\ \bibnamefont
  {Gisin}}, \ and\ \bibinfo {author} {\bibfnamefont {Lluis}\ \bibnamefont
  {Masanes}},\ }\bibfield  {title} {\enquote {\bibinfo {title} {From bell's
  theorem to secure quantum key distribution},}\ }\href {\doibase
  10.1103/PhysRevLett.97.120405} {\bibfield  {journal} {\bibinfo  {journal}
  {Phys. Rev. Lett.}\ }\textbf {\bibinfo {volume} {97}},\ \bibinfo {pages}
  {120405} (\bibinfo {year} {2006})}\BibitemShut {NoStop}%
\bibitem [{\citenamefont {Ac{\'{\i}}n}\ \emph {et~al.}(2006)\citenamefont
  {Ac{\'{\i}}n}, \citenamefont {Massar},\ and\ \citenamefont
  {Pironio}}]{AMP+06}%
  \BibitemOpen
  \bibfield  {author} {\bibinfo {author} {\bibfnamefont {Antonio}\ \bibnamefont
  {Ac{\'{\i}}n}}, \bibinfo {author} {\bibfnamefont {Serge}\ \bibnamefont
  {Massar}}, \ and\ \bibinfo {author} {\bibfnamefont {Stefano}\ \bibnamefont
  {Pironio}},\ }\bibfield  {title} {\enquote {\bibinfo {title} {Efficient
  quantum key distribution secure against no-signalling eavesdroppers},}\
  }\href {\doibase 10.1088/1367-2630/8/8/126} {\bibfield  {journal} {\bibinfo
  {journal} {New Journal of Physics}\ }\textbf {\bibinfo {volume} {8}},\
  \bibinfo {pages} {126--126} (\bibinfo {year} {2006})}\BibitemShut {NoStop}%
\bibitem [{\citenamefont {Maurer}\ and\ \citenamefont
  {Wolf}(1997)}]{MauWol97c-intr}%
  \BibitemOpen
  \bibfield  {author} {\bibinfo {author} {\bibfnamefont {Ueli}\ \bibnamefont
  {Maurer}}\ and\ \bibinfo {author} {\bibfnamefont {Stefan}\ \bibnamefont
  {Wolf}},\ }\bibfield  {title} {\enquote {\bibinfo {title} {The intrinsic
  conditional mutual information and perfect secrecy},}\ }in\ \href@noop {}
  {\emph {\bibinfo {booktitle} {Proc. 1997 IEEE Symposium on Information Theory
  (Abstracts)}}}\ (\bibinfo {year} {1997})\ p.~\bibinfo {pages}
  {88}\BibitemShut {NoStop}%
\bibitem [{\citenamefont {Maurer}\ and\ \citenamefont
  {Wolf}(1999)}]{Intrinsic-Maurer}%
  \BibitemOpen
  \bibfield  {author} {\bibinfo {author} {\bibfnamefont {U.}~\bibnamefont
  {Maurer}}\ and\ \bibinfo {author} {\bibfnamefont {S.}~\bibnamefont {Wolf}},\
  }\bibfield  {title} {\enquote {\bibinfo {title} {Unconditionally secure key
  agreement and the intrinsic conditional information},}\ }\href@noop {} {\
  \textbf {\bibinfo {volume} {45}},\ \bibinfo {pages} {499--514} (\bibinfo
  {year} {1999})}\BibitemShut {NoStop}%
\bibitem [{\citenamefont {Philip}\ \emph {et~al.}(2021)\citenamefont {Philip},
  \citenamefont {Kaur}, \citenamefont {Bierhorst},\ and\ \citenamefont
  {Wilde}}]{PKBW21}%
  \BibitemOpen
  \bibfield  {author} {\bibinfo {author} {\bibfnamefont {Aby}\ \bibnamefont
  {Philip}}, \bibinfo {author} {\bibfnamefont {Eneet}\ \bibnamefont {Kaur}},
  \bibinfo {author} {\bibfnamefont {Peter}\ \bibnamefont {Bierhorst}}, \ and\
  \bibinfo {author} {\bibfnamefont {Mark~M.}\ \bibnamefont {Wilde}},\
  }\bibfield  {title} {\enquote {\bibinfo {title} {Intrinsic non-locality and
  device-independent conference key agreement},}\ }\href@noop {} {\  (\bibinfo
  {year} {2021})},\ \bibinfo {note} {arXiv:2111.02596}\BibitemShut {NoStop}%
\bibitem [{\citenamefont {H{\"a}nggi}\ \emph {et~al.}(2010)\citenamefont
  {H{\"a}nggi}, \citenamefont {Renner},\ and\ \citenamefont
  {Wolf}}]{hanggi-2009}%
  \BibitemOpen
  \bibfield  {author} {\bibinfo {author} {\bibfnamefont {E.}~\bibnamefont
  {H{\"a}nggi}}, \bibinfo {author} {\bibfnamefont {R.}~\bibnamefont {Renner}},
  \ and\ \bibinfo {author} {\bibfnamefont {S.}~\bibnamefont {Wolf}},\
  }\bibfield  {title} {\enquote {\bibinfo {title} {Efficient quantum key
  distribution based solely on bell's theorem},}\ }\href@noop {} {\bibfield
  {journal} {\bibinfo  {journal} {EUROCRYPT}\ ,\ \bibinfo {pages} {216--234}}
  (\bibinfo {year} {2010})},\ \Eprint
  {http://arxiv.org/abs/arXiv.org:0911.4171} {arXiv.org:0911.4171} \BibitemShut
  {NoStop}%
\bibitem [{\citenamefont {Masanes}(2009)}]{masanes-2009-102}%
  \BibitemOpen
  \bibfield  {author} {\bibinfo {author} {\bibfnamefont {L.}~\bibnamefont
  {Masanes}},\ }\bibfield  {title} {\enquote {\bibinfo {title}
  {Universally-composable privacy amplification from causality constraints},}\
  }\href@noop {} {\bibfield  {journal} {\bibinfo  {journal} {Phys. Rev. Lett}\
  }\textbf {\bibinfo {volume} {102}},\ \bibinfo {pages} {140501} (\bibinfo
  {year} {2009})},\ \Eprint {http://arxiv.org/abs/arXiv.org:0807.2158}
  {arXiv.org:0807.2158} \BibitemShut {NoStop}%
\bibitem [{\citenamefont {Bowles}\ \emph {et~al.}(2016)\citenamefont {Bowles},
  \citenamefont {Francfort}, \citenamefont {Fillettaz}, \citenamefont
  {Hirsch},\ and\ \citenamefont {Brunner}}]{BFF+16}%
  \BibitemOpen
  \bibfield  {author} {\bibinfo {author} {\bibfnamefont {Joseph}\ \bibnamefont
  {Bowles}}, \bibinfo {author} {\bibfnamefont {J\'er\'emie}\ \bibnamefont
  {Francfort}}, \bibinfo {author} {\bibfnamefont {Mathieu}\ \bibnamefont
  {Fillettaz}}, \bibinfo {author} {\bibfnamefont {Flavien}\ \bibnamefont
  {Hirsch}}, \ and\ \bibinfo {author} {\bibfnamefont {Nicolas}\ \bibnamefont
  {Brunner}},\ }\bibfield  {title} {\enquote {\bibinfo {title} {Genuinely
  multipartite entangled quantum states with fully local hidden variable models
  and hidden multipartite nonlocality},}\ }\href {\doibase
  10.1103/PhysRevLett.116.130401} {\bibfield  {journal} {\bibinfo  {journal}
  {Physical Review Letters}\ }\textbf {\bibinfo {volume} {116}},\ \bibinfo
  {pages} {130401} (\bibinfo {year} {2016})}\BibitemShut {NoStop}%
\bibitem [{\citenamefont {Umegaki}(1962)}]{Ume62}%
  \BibitemOpen
  \bibfield  {author} {\bibinfo {author} {\bibfnamefont {Hisaharu}\
  \bibnamefont {Umegaki}},\ }\bibfield  {title} {\enquote {\bibinfo {title}
  {Conditional expectations in an operator algebra, {IV} (entropy and
  information)},}\ }\href {\doibase 10.2996/kmj/1138844604} {\bibfield
  {journal} {\bibinfo  {journal} {Kodai Mathematical Seminar Reports}\ }\textbf
  {\bibinfo {volume} {14}},\ \bibinfo {pages} {59--85} (\bibinfo {year}
  {1962})}\BibitemShut {NoStop}%
\bibitem [{\citenamefont {Winter}(2016)}]{Win16}%
  \BibitemOpen
  \bibfield  {author} {\bibinfo {author} {\bibfnamefont {Andreas}\ \bibnamefont
  {Winter}},\ }\bibfield  {title} {\enquote {\bibinfo {title} {Tight uniform
  continuity bounds for quantum entropies: conditional entropy, relative
  entropy distance and energy constraints},}\ }\href
  {https://doi.org/10.1007/s00220-016-2609-8} {\bibfield  {journal} {\bibinfo
  {journal} {Communications in Mathematical Physics}\ }\textbf {\bibinfo
  {volume} {347}},\ \bibinfo {pages} {291--313} (\bibinfo {year}
  {2016})}\BibitemShut {NoStop}%
\end{thebibliography}%
\end{document}